\documentclass[conference,citecolor=RubineRed]{IEEEtran}
\IEEEoverridecommandlockouts
\usepackage{cite}
\usepackage{amsmath,amssymb,amsfonts}
\usepackage{algorithmic}
\usepackage{graphicx}
\usepackage{textcomp}
\usepackage{xcolor}
\def\BibTeX{{\rm B\kern-.05em{\sc i\kern-.025em b}\kern-.08em
    T\kern-.1667em\lower.7ex\hbox{E}\kern-.125emX}}

\usepackage{amsmath}
\usepackage{amssymb}
\usepackage{amsthm}
  \theoremstyle{definition}
    \newtheorem{para}{}[section]
    \newtheorem{rem}[para]{Remark}
    \newtheorem{defn}[para]{Definition}
    \newtheorem{open}[para]{Open question}
    \newtheorem{example}[para]{Example}
   
  \theoremstyle{plain}
    
    \newtheorem{lemma}[para]{Lemma}
    \newtheorem{proposition}[para]{Proposition}
    \newtheorem{theorem}[para]{Theorem}
    \newtheorem{corollary}[para]{Corollary}

\usepackage{xcolor}
  \definecolor{RoyalBlue}{HTML}{0071BC}
  \definecolor{RubineRed}{HTML}{ED017D}
\usepackage{tikz}
\usetikzlibrary{positioning,arrows,calc}
\tikzset{%
  modal/.style={>=stealth?,shorten >=1pt,shorten <=1pt,auto,node distance=1.5cm, semithick}, 
  world/.style={circle,draw,minimum size=0.5cm,fill=gray!15},
  point/.style={circle,draw,inner sep=0.5mm,fill=black},
  reflexive above/.style={->,loop,looseness=7,in=120,out=60},
  reflexive below/.style={->,loop,looseness=7,in=240,out=300},
  reflexive left/.style={->,loop,looseness=7,in=150,out=210},
  reflexive right/.style={->,loop,looseness=7,in=30,out=330}
}

\usepackage[colorlinks=true,
            linktocpage=true
            citecolor=RubineRed,
            linkcolor=RoyalBlue,
            urlcolor=RubineRed]{hyperref}
\usepackage{subcaption}
\usepackage{proof}
\usepackage{multicol}
\usepackage{vwcol}
\usepackage{multirow}
\usepackage{listings}
\usepackage{doi}

\usepackage{marvosym}

\usepackage{fontawesome5}
\newcommand{\expl}{\scalebox{.8}{\text{\faBomb}}}

\def\BaseUrl{https://ianshil.github.io/CK}
\newcommand{\rocqdoc}[1]{\href{\BaseUrl/#1}{\raisebox{-0.8mm}{\includegraphics[height=0.83em]{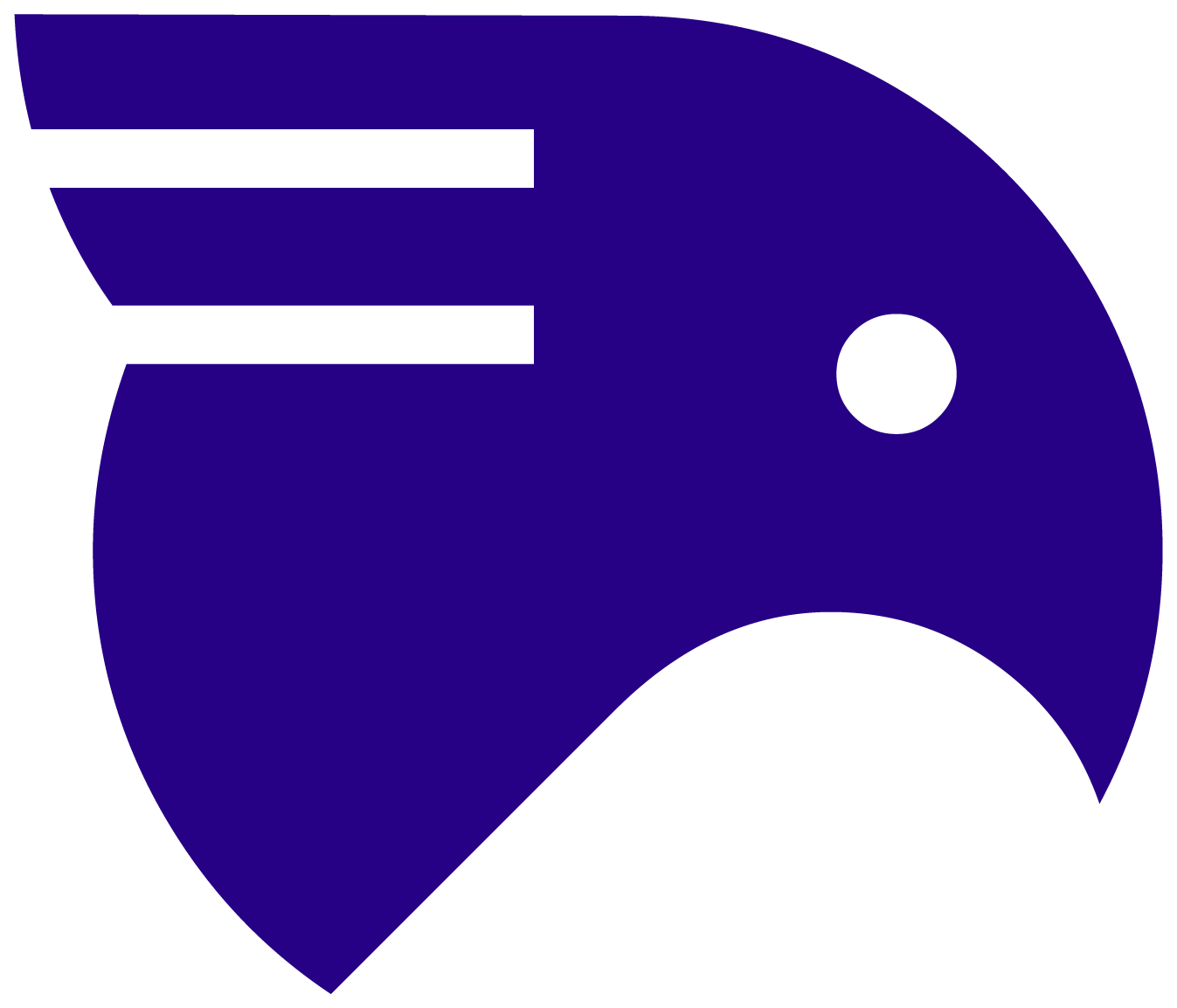}}}}

\newcommand{\lan}[1]{\mathbf{#1}}   % for languages
\newcommand{\deriv}[1]{\vdash_{\scriptstyle {#1}}}
\newcommand{\noderiv}[1]{\not\vdash_{\scriptstyle {#1}}}

\newcommand{\dneg}{\dot\neg}

\renewcommand{\phi}{\varphi}

\newcommand{\sement}[1]{\Vdash_{\scriptstyle{\mathcal{#1}}}}
\newcommand{\nosement}[1]{\not\Vdash_{\scriptstyle{\mathcal{#1}}}}

\newcommand{\mo}[1]{\mathfrak{#1}}
\renewcommand{\iff}{\quad\text{iff}\quad}
\renewcommand{\log}[1]{\mathsf{#1}}
\newcommand{\mc}[1]{\mathcal{#1}}

\newcommand{\rto}[1]{to node[#1]{\footnotesize{$R$}}}
\newcommand{\ito}[1]{to node[#1]{\footnotesize{ \ }}}

\newcommand{\subsetsim}{\mathrel{\substack{\textstyle\subset\\[-0.2ex]\textstyle\sim}}}

\newcommand{\Prop}{\mathbin{\mathrm{Prop}}}
\newcommand{\Ax}{\log{Ax}}
\newcommand{\CKAx}{\log{CKAx}}

\newcommand{\axref}[1]{\text{\ref{#1}}}

\newenvironment{myenumerate}
  {\medskip\begin{enumerate}}
  {\end{enumerate}\medskip}
  
\newcommand{\myitem}[1]{%
    \renewcommand{\labelenumi}{(\theenumi) }
    \renewcommand{\theenumi}{#1}
    \item%
  }

\usepackage{wasysym}
\let\oldDiamond\Diamond
\renewcommand{\Diamond}{%
  \mathchoice{\raisebox{-1pt}{$\displaystyle\oldDiamond$}}
             {\raisebox{-1pt}{$\oldDiamond$}}
             {\raisebox{-0.5pt}{$\scriptstyle\oldDiamond$}}
             {\raisebox{-0.2pt}{$\scriptscriptstyle\oldDiamond$}}}
\begin{document}

\title{Semantical Analysis of Intuitionistic Modal Logics between CK and IK
\thanks{
The authors would like to gratefully acknowledge discussions with Dirk Pattinson and Alwen Tiu, and
the detailed and helpful comments of the anonymous reviewers.
The second author has been supported by a UKRI Future Leaders Fellowship, 
‘Structure vs Invariants in Proofs’, project reference MR/S035540/1.}
}

\author{}

\author{\IEEEauthorblockN{Jim de Groot}
\IEEEauthorblockA{\textit{Mathematical Institute} \\
\textit{University of Bern}\\
Bern, Switzerland \\
https://orcid.org/0000-0003-1375-6758}
\and
\IEEEauthorblockN{Ian Shillito}
\IEEEauthorblockA{\textit{School of Computer Science} \\
\textit{University of Birimingham}\\
Birmingham, UK \\
https://orcid.org/0009-0009-1529-2679}
\and
\IEEEauthorblockN{Ranald Clouston}
\IEEEauthorblockA{\textit{School of Computing} \\
\textit{Australian National University}\\
Canberra, Australia \\
ranald.clouston@anu.edu.au}
}

\maketitle

\begin{abstract}
  The intuitionistic modal logics considered between Constructive K (CK) and
  Intuitionistic K (IK) differ in their treatment of the possibility
  (diamond) connective.
  It was recently rediscovered that some logics between CK and IK also
  disagree on their diamond-free fragments, with only some
  remaining conservative over the standard axiomatisation of intuitionistic
  modal logic with necessity (box) alone.
  We show that relational Kripke semantics for CK can be extended with frame
  conditions for all axioms in the standard axiomatisation of IK, as well
  as other axioms previously studied.
  This allows us to answer open questions about the (non-)conservativity of
  such logics over intuitionistic modal logic without diamond.
  Our results are formalised using the Rocq Prover.
\end{abstract}


%%%%%%%%%%%%%%%%%%%%%%%%%%%%%%%%%%%%%%%%%%%%%%%%%%%%%%%%%%%%%%%%%%%%%%%%%%%%%%%%
%%%%%%%%%%%%%%%%%%%%%%%%%%%%%%%%%%%%%%%%%%%%%%%%%%%%%%%%%%%%%%%%%%%%%%%%%%%%%%%%
\section{Introduction}

Which logic provides the foundation for intuitionistic modal logics, by analogy with the logic $\log{K}$ for classical modal logics? If we consider
necessity, $\Box$, but disregard possibility, $\Diamond$, then the answer has, until recently, appeared uncontroversial: we extend intuitionistic
propositional logic with the inference rule of necessitation (if $p$ is a theorem, then so is $\Box p$) and axiom
  \begin{myenumerate}
    \setlength{\itemindent}{1em}
    \myitem{$\mathsf{K_{\Box}}$} \label{ax:Kb}
          $\Box(\phi \to \psi) \to (\Box \phi \to \Box \psi)$
  \end{myenumerate}
  Since this logic, which we here call $\log{CK}_{\Box}$, was introduced by
  Bo\v{z}i\'{c} and Do\v{s}en~\cite{BozDos84}, it and its extensions have been
  studied and applied in a literature too large to summarise here;
  some examples are given in Section~\ref{sec:motivation-box}.

How $\log{CK}_{\Box}$ should be extended with $\Diamond$ has received various answers. Consider the following axioms, where we follow the
naming conventions of Dalmonte, Grellois, and Olivetti~\cite{DalGreOli20}:
  \begin{myenumerate}
    \setlength{\itemindent}{1em}
    \myitem{$\mathsf{K_{\Diamond}}$} \label{ax:Kd}
          $\Box(\phi \to \psi) \to (\Diamond \phi \to \Diamond \psi)$
    \myitem{$\mathsf{N_{\Diamond}}$} \label{ax:Nd}
          $\Diamond\bot \to \bot$
    \myitem{$\mathsf{C_{\Diamond}}$} \label{ax:Cd}
          $\Diamond(\phi \vee \psi) \to \Diamond \phi \vee \Diamond \psi$
    \myitem{$\mathsf{I_{\Diamond\Box}}$} \label{ax:Idb}
          $(\Diamond \phi \to \Box \psi) \to \Box(\phi \to \psi)$
  \end{myenumerate}
\noindent
Constructive $\log{K}$ ($\log{CK}$)~\cite{BeldePRit01} extends $\log{CK}_{\Box}$ with \axref{ax:Kd} only; Wijesekera's $\log{K}$ ($\log{WK}$)~\cite{Wij90}
extends $\log{CK}$ further with \axref{ax:Nd}. These logics are proof theoretically natural, attained by restricting the sequent calculus for classical
$\log{K}$ to single conclusions ($\log{CK}$) or zero or one conclusions ($\log{WK}$). Both were originally also motivated by applications in AI: the
notion of context in knowledge representation and reasoning ($\log{CK}$)~\cite{Pai03,MendeP05}, and representing states with partial knowledge, as well as
constructive concurrent dynamic logic ($\log{WK}$)~\cite{WijNer05},
as discussed further in Section~\ref{sec:motivation}. Intuitionistic $\log{K}$ ($\log{IK}$)~\cite{Fis84}%
\footnote{The name $\log{IK}$ has been used inconsistently in the literature, sometimes for the logics that we here call $\log{CK_\Box}$ and
$\log{CK}$.}
has all the above axioms. It is the logic specified by Fischer Servi's translation to classical $\log{(K,S4)}$-bimodal logic~\cite{Fis77,Fis81,Fis84}, and by the standard
translation to intuitionistic first order logic~\cite{Sim94}.
$\log{IK}$ also respects the  G\"{o}del-Gentzen double negation translation from classical modal
logic, although this also holds for the logic without $\axref{ax:Cd}$~\cite{DasMar23}. These are by no means the only options for logics between
$\log{CK}$ and $\log{IK}$; we mention also Kojima's logic for intuitionistic neighbourhood models~\cite{Koj12}, which lies between $\log{CK}$ and
$\log{WK}$, and Forward confluence $\log{IK}$ ($\log{FIK}$)~\cite{BalGaoGenOli24}, which modifies $\log{IK}$ by replacing \axref{ax:Idb} with
a weaker axiom.

While different notions of $\Diamond$ have arisen from different motivations, it has generally been assumed that only $\Diamond$ is controversial,
and that these logics agree with $\log{CK}_{\Box}$ on their $\Diamond$-free fragments. This was shown to be incorrect in Grefe's 1999
thesis~\cite{Gre99}. Grefe showed that the $\Diamond$-free formula $(\lnot\Box\bot\to\Box\bot)\to\Box\bot$ holds in $\log{IK}$ but not in
$\log{CK}_{\Box}$. This observation was not published, and was only recently rediscovered by Das and Marin~\cite{DasMar23}, who showed,
among other results, that while $\log{IK}$ is not conservative over $\log{CK}_{\Box}$, the logic $\log{CK}\oplus\axref{ax:Nd}\oplus\axref{ax:Cd}$,
and hence its sublogics such as $\log{CK}$ and $\log{WK}$, are.

The (re)discovery that the logics between $\log{CK}$ and $\log{IK}$ are not as well understood as previously thought raises many questions.
With each new axiom that is considered in this space, these questions multiply. Working via Hilbert axiomatisations only is notoriously intractable, and
while proof theoretic methods were used successfully by Das and Marin to clarify the status of $\log{CK}\oplus\axref{ax:Nd}\oplus\axref{ax:Cd}$,
the effort involved was considerable. In this paper we instead explore the (bi)relational semantics of Kripke frames.
Such semantics are known for $\log{CK}$~\cite{MendeP05}, $\log{WK}$~\cite{Wij90}, $\log{FIK}$~\cite{BalGaoGenOli24}, and $\log{IK}$~\cite{Fis81},
although imprecisions in the treatment of $\log{WK}$ led to the soundness proof with respect to the semantics for that logic being called ``inconclusive''~\cite{MendeP05}%
\footnote{More precisely, Mendler and De Paiva argued that flaws in Wijesekera's work with WK made it
unsuitable to conclude soundness for CK, but their argument holds equal force as a criticism of the development
for WK itself.}%
. Moreover, the different choices made both for conditions on the relations and for the interpretations of the modal
operators impede comparisons between these logics.

In this paper we take the relational semantics for $\log{CK}$ as a unifying semantics, and give frame conditions for each
of the axioms \axref{ax:Nd}, \axref{ax:Cd}, and \axref{ax:Idb}. This allows us to provide
completeness proofs for each of these axioms independently. In particular, this answers the challenge of Das and Marin~\cite[Section~7]{DasMar23} to
provide relational semantics for logics between $\log{WK}$ and $\log{IK}$. We use these semantics to analyse the $\Diamond$-free fragments,
making the new observations that $\log{CK}\oplus\axref{ax:Cd}\oplus\axref{ax:Idb}$, and hence $\log{CK}\oplus\axref{ax:Idb}$, are
conservative over $\log{CK_{\Box}}$. This is a surprising result, as no logics including \axref{ax:Idb} were previously shown to retain
conservativity; it is now clear that the \emph{combination} of \axref{ax:Idb} and \axref{ax:Nd} is to blame here.

We formalise all our results in the
Rocq Prover~\cite{Coq}, which not only adds confidence to our results 
(in particular, the doubt raised~\cite{MendeP05} about the relational semantics for $\log{WK}$ may now be considered settled), 
but is a crucial working tool for managing the profusion of logics which arise as one considers new axioms. As a proof of concept
of this methodology of working from a base relational semantics for $\log{CK}$ with support from Rocq, we go on to provide relational semantics and
conservativity results for Kojima's logic,
and for the weakening of \axref{ax:Idb} used in $\log{FIK}$.
  Each mechanised result in the paper is accompanied by a clickable rooster symbol 
  ``\raisebox{-.9mm}{\includegraphics[height=1em]{coql.png}}'' leading to
  its mechanisation.
  The full mechanisation can be found at \url{https://github.com/ianshil/CK}
  and its documentation at \url{https://ianshil.github.io/CK/toc.html}.

This paper begins by discussing constructive modal logics in more depth
in Section~\ref{sec:motivation}, before introducing the basic syntax in Section~\ref{sec:consK}. We give sound and strongly complete relational semantics for our base
logic $\log{CK}$ in Section~\ref{sec:rel-sem-CK}, then give frame conditions and completeness results for \axref{ax:Nd}, \axref{ax:Cd} and \axref{ax:Idb} in
Sections~\ref{sec:three-axioms} and~\ref{sec:eight}, and compare the resulting logics' $\Diamond$-free fragments in Section~\ref{sec:diamond-free}. We extend our
techniques to other axioms from the literature in Section~\ref{sec:other}. We finish by discussing our Rocq formalisation in Section~\ref{sec:formal} and
surveying possible further work in Section~\ref{sec:further}.
%

%%%%%%%%%%%%%%%%%%%%%%%%%%%%%%%%%%%%%%%%%%%%%%%%%%%%%%%%%%%%%%%%%%%%%%%%%%%%%%%%
\section{Constructive K and its extensions}\label{sec:motivation}

  In this section we elaborate on the logics that lie at the heart of the paper.

%-------------------------------------------------------------------------------
\subsection{Intuitionistic logic with boxes}\label{sec:motivation-box}

  Intuitionistic modal logics with a necessity (box) modality but no diamond,
  or where diamond is viewed as a derived modality, date back to
  1965~\cite{Bul65b}. Early literature often takes an $\log{S4}$-perspective
  on the modality~\cite{Pra65,Ono77,Fon86,BiePai96}.
  The first occurrence of $\log{CK}_{\Box}$ appears to be in 1984~\cite{BozDos84,Dos85},
  where it is called $\boldsymbol{H\!K\Box}$. Subsequently, it has been widely
  studied under various names, including
  $\mathbf{IntK}$~\cite{WolZak98}, $\mathbf{IntK_{\Box}}$~\cite{WolZak97,WolZak99},
  $\mathbf{IK_{\Box}}$~\cite{BeldePRit01}, $\mathbf{IK}$~\cite{Kak07}
  and $iK$~\cite{DasMar23}.
  We highlight some of its appearances.

\begin{example}[Modalities for context]
  The modal operator $\Box$ can be used to formalise the idea of a \emph{context}, a notion in the field of \emph{knowledge representation}.
  For example, if $\Box$ denotes the context of Sherlock Holmes, then it is true that
  Sherlock Holmes lives in Baker Street, i.e.~$\Box(\text{Sherlock lives on Baker Street})$.
  We can use multiple modalities, denoted as
  $\Box_{\kappa}$ or $\texttt{ist}(\kappa, \phi)$ (for \emph{is t}rue), to model several contexts $\kappa$.

  From a computer science point of view, contexts can for example be used to deal
  with databases with multiple conventions~\cite{McCBuv94,McC96}.
  More generally McCartey states that an ``AI goal'' is to allow
  simple axioms for commonsense to be lifted to other contexts~\cite{McC93,McCBuv94}.
  This idea was further studied in e.g.~\cite{BuvBuvMas95,Nay94,Mas95}, and
  in~\cite{Pai03} it was shown that the common core of the latter three is
  given by (a multimodal version of) $\log{CK}_{\Box}$.
\end{example}

\begin{example}[Modalities for knowledge]
  An epistemic interpretation of $\Box\phi$ is that an agent \emph{knows}
  or \emph{believes} $\phi$ to be true.
  In an intuitionistic context, the epistemic operator can be used to model
  an ideal reasoner (the agent) in a growing informational state (an intuitionistic Kripke frame)~\cite{Pro12,JagMar16}.
  This motivates the reflection principle $\Box\phi \to \phi$: if an agent
  knows that $\phi$ is true, then it is true.
  
  Alternatively, one can take $\Box\phi$ to represent ``belief and knowledge as the product of verification''~\cite{ArtPro16}.
  In this view, the intuitionistic truth of a proposition entails knowledge of it,
  because an intuitionistic proof is a verification, so one gets
  the co-reflection principle $\phi \to \Box\phi$ as an axiom.
  A priori, this logic does not rule out false beliefs.
  The extension of $\log{CK}_{\Box}$ with co-reflection is called 
  $\log{IEL}^-$, and has recently received a lot of attention~\cite{Rog20,Rog20b,Bro21,SuSan23}.
  (Incidentally, $\log{IEL}^-$ coincides with
  the inhabitation logic of Haskell's applicative functors~\cite{McBPat08},
  as was noted in~\cite{LitPolRab17,Rog20}.)
\end{example}

\begin{example}[Curry-Howard correspondence]
  Constructive versions of $\log{S4}$ received a lot of attention from a
  type-theoretic perspective~\cite{GhaPaiRit98,BiePai00,AleMenPaiRit01}.
  This sparked attempts to give a
  Curry-Howard correspondence for $\log{CK}_{\Box}$ as well.
  The first such correspondence was established by Bellin, De Paiva and Ritter~\cite{BeldePRit01}, and was later refined by Kakutani~\cite{Kak07}. In their work, $\Box$ is the type former corresponding to a term constructor which can be interpreted as a sort of substitution. This is still an active field of research: a new correspondence for the $\land\lor$-free fragment of $\log{CK}_{\Box}$ was recently discovered by Acclavio, Catta and Olimpieri~\cite{AccCatOli23}, and a Curry-Howard correspondence for $\log{IEL}^-$
  was given in~\cite{Rog20b,Bro21}.
\end{example}

  Extensions of $\log{CK}_{\Box}$, for example with the $\log{S4}$ axioms,
  have also found many applications, ranging from hardware verification~\cite{FaiMen97} to
  access control~\cite{GarPfe06} to
  staged computation~\cite{DavPfe96,DavPfe01,NanPfePie08},
  and from the productivity of recursive definitions~\cite{BizGraCloMoeBir16}
  to global elements in synthetic topology~\cite{Shu18}.

%-------------------------------------------------------------------------------
\subsection{Intuitionistic logic with boxes and diamonds}

  As in the mono-modal case, the study of intuitionistic modal logic with two
  modalities, $\Box$ and $\Diamond$, started with intuitionistic
  analogues of $\log{S4}$, for example in~\cite{Pri57,Bul65,Pra65,Bul66,PfeDav01}.
  These were then generalised to intuitionistic counterparts of $\log{K}$,
  where the variety of axioms defining $\Diamond$ and relating
  $\Box$ and $\Diamond$ (such as \axref{ax:Nd}, \axref{ax:Cd} and \axref{ax:Idb})
  resulted in a wide variety of intuitionistic modal logics.
  
  One of the simplest intuitionistic modal logics with an independent
  box and diamond modality is Constructive $\log{K}$ ($\log{CK}$).
  This extends $\log{CK}_{\Box}$ with $\axref{ax:Kd}$,
  and was described in~\cite{BeldePRit01},
  following an adaptation of Prawitz's suggestions~\cite{Pra65}
  for intuitionistic $\log{S4}$ to $\log{K}$.
  Adding various configurations of~\axref{ax:Nd}, \axref{ax:Cd} and \axref{ax:Idb}
  gives rise to logics including $\log{WK}$~\cite{Wij90,WijNer05},
  and $\log{IK}$~\cite{Fis81,Fis84,PloSti86,Sim94}.
  We point out some uses of these logics.

\begin{example}[Satisfiability in context]
  In the setting of knowledge representation, $\Diamond_{\kappa}\phi$
  can be can be interpreted as $\phi$ being \emph{satisfiable in context}
  $\kappa$~\cite{MendeP05}. 
  Under this light, the diamond-containing axiom \axref{ax:Kd} of $\log{CK}$ is a sensible one to adopt. Indeed, truth of the implication of $\phi\to\psi$ in a given context allows one to infer the satisfiability of $\psi$ from satisfiability of $\phi$.
  However, we may not so readily accept other axioms, like \axref{ax:Nd} and \axref{ax:Cd}. 
  For example, \axref{ax:Nd} declares that falsity is satisfiable in no context,
  so adding it to our system prevents us from identifying inconsistent contexts.
\end{example}

\begin{example}[Parallel computation]
  The logic $\log{WK}$ is obtained by adding~\axref{ax:Nd} to $\log{CK}$~\cite{Wij90,WijNer05}.
  It was put forward as a constructivised version of concurrent dynamic logic~\cite{Pel87}.
  Here $\Diamond_{\alpha}\phi$ means that an execution of program $\alpha$
  reaches a state where $\phi$ holds,
  and $\Diamond_{\alpha \cap \beta}$ is read as
  ``$\alpha$ and $\beta$ can be executed in parallel so that upon termination
  (in either computation path) $\phi$ holds,''
  so it is equivalent to $\Diamond_{\alpha}\phi \wedge \Diamond_{\beta}\phi$.
  This interpretation prevents distributivity of diamonds over joins (i.e.~\axref{ax:Cd}),
  because the truth of $\Diamond_{\alpha \cap \beta}(\phi \vee \psi)$ may be
  witnessed by $\Diamond_{\alpha}\phi$ and $\Diamond_{\beta}\psi$.
\end{example}

\begin{example}[Diamonds for consistency]
  Both classically~\cite{JaaHin62} and intuitionisticially~\cite{Wil92},
  the diamond operator is used in epistemic logic to denote
  consistency or a kind of possibility of $\phi$ with respect to an agent's knowledge.
  The disentanglement of box and diamond in the intuitionistic setting allows
  us to reevaluate the axioms we impose on diamonds.
  
  For example, if $\phi \to \psi$ is known and $\phi$ is possible (or consistent),
  then it stands to reason that $\psi$ is consistent too,
  so~\axref{ax:Kd} is a plausible axiom.
  The axiom~\axref{ax:Nd} holds in the intuitionistic epistemic logic studied in~\cite{Wil92},
  but we may not always want this to be the case:
  Since the point of view taken in~\cite{ArtPro16} allows an agent
  to hold a false belief, $\bot$ could be a consequence of their knowledge,
  so that $\Diamond\bot$ holds and we must reject~\axref{ax:Nd}.
\end{example}

\begin{example}[Evaluation logic]
  In~\cite{Pit91}, Pitts introduces \emph{evaluation logic}, which is an
  extension of $\log{IK}$. This logic has modal formulas of the form
  $[x \Leftarrow E]\phi(x)$ and $\langle x \Leftarrow E \rangle \phi(x)$,
  which express that if $x$ is evaluated to $E$, then $\phi(x)$ will
  necessarily or possibly hold.
  The logic is designed to reason about computation specified using
  a style of operational semantics called natural semantics.
\end{example}

\begin{example}[Curry-Howard correspondence]
  It is natural to wonder whether the Curry-Howard correspondence for
  $\log{CK}_{\Box}$ can be extended to one of the above-mentioned constructive
  modal logics with a diamond.
  A correspondence for $\log{CK}$ was given in~\cite{BeldePRit01},
  but this turned out to have deficiencies~\cite{Kak07,PaiRit11,Kav16}
  which as of yet have not been entirely corrected.

  A correspondence for $\log{WK}$ would be particularly attractive, given its
  interpretation as parallel computation. This would allow one to generate
  programs containing concurrency which are verified by extraction.
\end{example}

%%%%%%%%%%%%%%%%%%%%%%%%%%%%%%%%%%%%%%%%%%%%%%%%%%%%%%%%%%%%%%%%%%%%%%%%%%%%%%%%
\section{The formal system(s)}\label{sec:consK}
  
  In this section we fix the syntax and axiomatic calculus for $\log{CK}$
  and its extensions.
  Taking a countably infinite set of propositional variables~$\Prop=\{p,q,r, \dots\}$, we define
  the language $\lan{L}$ via the following grammar (\rocqdoc{Syntax.im_syntax.html\#form}):
  \begin{equation*}
    \varphi ::= p\in\Prop \mid \bot \mid \varphi \land \varphi \mid \varphi
    \lor \varphi \mid \varphi \rightarrow \varphi \mid \Box\varphi \mid \Diamond\varphi
  \end{equation*}
  We abbreviate $\neg\phi:=\phi\rightarrow\bot$ and $\top:=\neg\bot$.
  We use Greek lowercase letters, e.g. $\phi, \psi, \chi$ and $\delta$, to
  denote formulas, and Greek uppercase letters, e.g. $\Gamma, \Delta, \Phi, \Psi$,
  for multisets of formulas.
  For such a multiset $\Gamma$ we define the multisets
  $\Box(\Gamma) := \{\Box\varphi \mid \varphi\in\Gamma\}$
  and $\Box^{-1}(\Gamma) := \{ \phi \mid \Box\phi \in \Gamma \}$ and similarly for
  $\Diamond(\Gamma)$ and $\Diamond^{-1}(\Gamma)$.
  If $\Gamma$ is finite, $\bigvee\Gamma$ denotes the disjunction of
  all formulas in $\Gamma$.
  We distinguish the logical connectives in $\lan{L}$ from those used in our metalogic
  with a dot on top of the metalogical connectives,
  e.g.~$\dneg$.
  Since $\Prop$ is countably infinite and we have finitely
  many connectives we can enumerate the formulas of $\lan{L}$ (\rocqdoc{GHC.enum.html\#form_enum}).

  All logics we consider are syntactically defined as
  extensions of the base logic $\log{CK}$ with axioms.
  We describe this formally by defining a logic parametrised in a
  set $\Ax \subseteq \mathbf{L}$ of axioms, so that $\Ax = \emptyset$
  corresponds to $\log{CK}$.\footnote{In the formalisation
  we use as parameter a set of formulas closed under substitution. 
  Given a set $\Ax$ of axioms, the set of all instances of axioms in $\Ax$
  is such a set.}
  We denote by $\mathcal I(\Ax)$ the set of all instances of axioms
  in a given set $\Ax$.

\begin{defn}[\rocqdoc{GHC.CKH.html\#extCKH_prv}]
  Let $\log{\CKAx}$ (\rocqdoc{GHC.CKH.html\#Axioms}) be an axiomatisation
  of intuitionistic logic (\rocqdoc{GHC.CKH.html\#IAxioms})
  together with $\axref{ax:Kb}$ and $\axref{ax:Kd}$.
  For a set $\Ax \subseteq \mathbf{L}$, % of axioms,
  define the generalised Hilbert calculus $\log{CK \oplus \Ax}$ by:
  \medskip
  \begin{multicols}{2}%
  \begin{enumerate}\itemsep=6pt
    \setlength{\itemindent}{1.1em}
    \renewcommand{\labelenumi}{(\theenumi) }
    %% first column
    \renewcommand{\theenumi}{$\mathsf{Ax}$}
    \item \label{rule:Ax}
          $\dfrac{\phi\in\mathcal{I}(\CKAx)\cup \mathcal I (\Ax)}{\Gamma \vdash \phi}$
    \renewcommand{\theenumi}{$\mathsf{MP}$}
    \item \label{rule:MP}
          $\dfrac{\Gamma \vdash \phi \qquad \Gamma \vdash \phi \to \psi}{\Gamma \vdash \psi}$
    %% second column
    \setlength{\itemindent}{2.1em}
    \renewcommand{\theenumi}{$\mathsf{Nec}$}
    \item \label{rule:Nec}
          $\dfrac{\emptyset \vdash \phi}{\Gamma \vdash \Box\phi}$
    \renewcommand{\theenumi}{$\mathsf{El}$}
    \item \label{rule:El}
          $\dfrac{\phi \in \Gamma}{\Gamma \vdash \phi}$
  \end{enumerate}
  \end{multicols}
  
\medskip\noindent
  We call \emph{consecutions} expressions of the form $\Gamma\vdash\phi$. 
  We say that $\Gamma\vdash\varphi$ is \emph{provable in} $\log{CK}\oplus\Ax$, and write $\Gamma\deriv{\Ax}\varphi$, if
  there exists a tree of consecutions built using the rules above with $\Gamma\vdash\varphi$ as root and adequate applications of rules \ref{rule:El} and \ref{rule:Ax} as leaves.
  We also write $\Gamma\noderiv{\Ax}\varphi$ if $\dneg(\Gamma\deriv{\Ax}\varphi)$, and write $\Gamma\deriv{\Ax}\Delta$ for $\Delta\subseteq\lan{L}$ if
  there is a finite $\Delta'\subseteq\Delta$ such that $\Gamma\deriv{\Ax}\bigvee\Delta'$.
  If $\Ax=\{\log{A_0},\dots,\log{A_n}\}$ is finite, we write
  $\log{CK} \oplus \log{A_0} \oplus \dots \oplus \log{A_n}$ for $\log{CK} \oplus \Ax$.
\end{defn}

  Sometimes $\log{CK} \oplus \Ax$ has an existing name in the literature.
  For example, $\log{CK} \oplus \axref{ax:Nd}$ is known as $\log{WK}$.
  In such cases, we use both names interchangeably.
The rules displayed below, where $\sigma$ is a uniform substitution, are admissible
in $\log{CK \oplus\Ax}$ for any set of axioms $\Ax$:
\begin{center}
\begin{tabular}{c@{\hspace{1cm}}l@{\hspace{1cm}}l}
$
\inferLineSkip=3pt
\infer{\Gamma,\Gamma'\vdash\varphi}{\Gamma\vdash\varphi}
$ & 
$
\inferLineSkip=3pt
\infer{\Gamma\vdash\varphi}{
	\{\Gamma\vdash\delta \; \mid \; \delta\in\Delta\}
	&
	\Delta\vdash\varphi}
$ & 
$
\inferLineSkip=3pt
\infer{\Gamma^\sigma\vdash\varphi^\sigma}{\Gamma\vdash\varphi}
$\\
\end{tabular}
\end{center}
\begin{center}
\begin{tabular}{c@{\hspace{2cm}}c}
$
\inferLineSkip=3pt
\infer={\Gamma\vdash\phi\rightarrow\psi}{\Gamma,\phi\vdash\psi}
$ & 
$
\inferLineSkip=3pt
\infer{\Box(\Gamma)\vdash\Box\varphi}{\Gamma\vdash\varphi}
$  \\
\end{tabular}
\end{center}
The three topmost rules show that $\log{CK} \oplus \Ax$ is a monotone (\rocqdoc{GHC.logic.html\#extCKH_monot}),
compositional (\rocqdoc{GHC.logic.html\#extCKH_comp}) and
structural (\rocqdoc{GHC.logic.html\#extCKH_struct}) relation, respectively.
Furthermore, we can show that $\Gamma \deriv{\Ax} \phi$ if and only if
there is a finite $\Gamma' \subseteq \Gamma$ such that $\Gamma' \deriv{\Ax} \phi$
(\rocqdoc{GHC.logic.html\#extCKH_finite}).
Therefore $\log{CK} \oplus \Ax$ is a finitary logic~\cite{Kra99}.
The left rule of the bottom row can be applied in both directions and
corresponds to the deduction-detachment theorem (\rocqdoc{GHC.properties.html\#extCKH_Deduction_Theorem},\rocqdoc{GHC.properties.html\#extCKH_Detachment_Theorem}).%
\footnote{This notably implies that $\log{CK}$ and $\log{WK}$ satisfy the deduction theorem.
Mendler and De Paiva make the opposite claim~\cite[Footnote 2]{MendeP05} in their analysis of Wijesekera's work,
but despite some imprecision in his definitions, we believe that Wijesekera had in mind a calculus like ours where
the rule (\axref{rule:Nec}) has $\emptyset$, and not a general $\Gamma$, in its premise.}
The right rule of the bottom row captures the modal sequent calculus rule (\rocqdoc{GHC.properties.html\#K_rule}).

\begin{defn}\label{def:theory}
  A set of formulas $\Gamma \subseteq \mathbf{L}$ is a \emph{theory} (\rocqdoc{GHC.Lindenbaum_lem.html\#closed})
  if it is deductively closed, i.e.\ $\Gamma \deriv{\Ax} \phi$ implies $\phi \in \Gamma$.
  It is \emph{prime} (\rocqdoc{GHC.Lindenbaum_lem.html\#prime}) if $\phi \vee \psi \in \Gamma$ implies
  $\phi \in \Gamma$ or $\psi \in \Gamma$, for all $\phi, \psi \in \mathbf{L}$.
\end{defn}

  We note that the prime theories we consider need not 
  be consistent, and can thus contain $\bot$. This reflects the existence
  of an inconsistent world in the semantics, defined in Section~\ref{sec:rel-sem-CK}.

\begin{lemma}[Lindenbaum \rocqdoc{GHC.Lindenbaum_lem_pair.html\#Lindenbaum_pair}]\label{lem:Lind}
Let $\Gamma\cup\Delta\subseteq\lan{L}$.
If $\Gamma\noderiv{\Ax}\Delta$ then there is a prime theory $\Gamma'\supseteq\Gamma$
such that $\Gamma'\noderiv{\Ax}\Delta$.
\end{lemma}

This can be proved by a routine argument.
We often use the Lindenbaum lemma with $\Delta$ of the form $\{ \phi \}$~(\rocqdoc{GHC.Lindenbaum_lem.html\#Lindenbaum}).

%%%%%%%%%%%%%%%%%%%%%%%%%%%%%%%%%%%%%%%%%%%%%%%%%%%%%%%%%%%%%%%%%%%%%%%%%%%%%%%%
\section{Relational semantics for $\log{CK}$}\label{sec:rel-sem-CK}

  We present a relational semantics for $\log{CK}$ which is a light
  modification of Mendler and De Paiva~\cite{MendeP05}.
  Their semantics is characterised by
  the interpretation of \emph{both} modalities over all intuitionistic successors of worlds, and by
  the existence of worlds that satisfy all formulas, including $\bot$.
  Such worlds were introduced by Veldman~\cite{Vel76} as ``sick'' worlds,
  whereas Mendler and De Paiva call them ``fallible'';
  we follow Ilik, Lee and Herbelin's \emph{exploding} terminology~\cite{IliLeeHer10}.
  Because all exploding worlds are essentially the same with respect to formula satisfaction
  we slightly simplify the Mendler-De Paiva semantics by using a single exploding world,
  instead of a set.

\begin{defn}
  A \emph{$\log{CK}$-frame} (\rocqdoc{Kripke.kripke_sem.html\#frame})
  is a tuple $(X, \expl, \leq, R)$ where $(X, \leq)$ is a preorder,
  $\expl \in X$ is a maximal element of $(X, \leq)$,
  and $R$ is a binary relation on $X$ such that $\expl R x$ if and only if $x = \expl$.
  We denote by $\mathcal{CK}$ the class of all $\log{CK}$-frames.
  
  A \emph{valuation} is a map $V$ that assigns to each proposition letter $p$
  an upset $V(p)$ of $(X, \leq)$ such that $\expl \in V(p)$.
  A \emph{$\log{CK}$-model} (\rocqdoc{Kripke.kripke_sem.html\#model}) is a $\log{CK}$-frame with a valuation.
  The interpretation of a formula $\phi$ at a world $x$ in a $\log{CK}$-model
  $\mo{M} = (X, \expl, \leq, R, V)$ (\rocqdoc{Kripke.kripke_sem.html\#forces}) is defined recursively by
  \begin{align*}
    \mo{M}, x \Vdash p &\iff x \in V(p) \\
    \mo{M}, x \Vdash \bot &\iff x = \expl \\
    \mo{M}, x \Vdash \phi \wedge \psi
              &\iff \mo{M}, x \Vdash \phi \text{ and } \mo{M}, x \Vdash \psi \\
    \mo{M}, x \Vdash \phi \vee \psi
              &\iff \mo{M}, x \Vdash \phi \text{ or } \mo{M}, x \Vdash \psi \\
    \mo{M}, x \Vdash \phi \to \psi
              &\iff \forall y \; ( x \leq y
                                    \text{ and } \mo{M}, y \Vdash \phi \\
                                    &\hspace{7em}\text{ imply } \mo{M}, y \Vdash \psi )\\
    \mo{M}, x \Vdash \Box\phi
            &\iff \forall y, z \; (x \leq y
                                     \text{ and } yRz \\
                                     &\hspace{7em}\text{ imply } \mo{M}, z \Vdash \phi) \\
    \mo{M}, x \Vdash \Diamond\phi
            &\iff \forall y \; (x \leq y
                            \text{ implies } \exists z \in X \\
                            &\hspace{7em}\text{ s.t. } yRz \text{ and } \mo{M}, z \Vdash \phi)
  \end{align*}
  Let $\Gamma \cup \{ \phi \} \subseteq \lan{L}$
  and let $\mo{M}$ be a $\log{CK}$-model.
  We write $\mo{M}, x \Vdash \Gamma$ if $x$ satisfies all $\psi \in \Gamma$,
  and we say that $\mo{M}$ \emph{validates} $\Gamma \vdash \phi$ if
  $\mo{M}, x \Vdash \Gamma$ implies $\mo{M}, x \Vdash \phi$ for all worlds
  $x$ in $\mo{M}$.
  A $\log{CK}$-frame $\mo{X}$ \emph{validates} $\Gamma \vdash \phi$ if every model
  of the form $(\mo{X}, V)$ validates the consecution, and it validates a
  formula $\phi$ if it validates the consecution $\emptyset \vdash \phi$.
  If $\mathcal{F}$ is a class of $\log{CK}$-frames, then we say that
  \emph{$\Gamma$ semantically entails $\phi$ on $\mathcal{F}$}
  (\rocqdoc{Kripke.kripke_sem.html\#loc_conseq}),
  and write $\Gamma \Vdash_{\mathcal{F}} \phi$,
  if every $\log{CK}$-frame in $\mathcal{F}$ validates $\Gamma \vdash \phi$.
\end{defn}

  The universal quantifier in the interpretation of $\Diamond$ prevents
  distributivity of diamond over disjunctions, and thus is often not
  necessary when studying logics that include~$\axref{ax:Cd}$.
  We reiterate that $\expl$ is a maximal element in $(X, \leq)$ but not
  necessarily a top element. That is, there are no elements above $\expl$
  in the partial order $\leq$ (other than $\expl$ itself), but $\expl$ does not
  necessarily lie above all elements of $X$.

  We will show that $\log{CK}$-frames form a sound and complete semantics for $\log{CK}$.

\begin{lemma}[Persistence \rocqdoc{Kripke.kripke_sem.html\#Persistence}]
  If $\mo{M}, x \Vdash \phi$ and $x \leq y$ then $\mo{M}, y \Vdash \phi$.
\end{lemma}
\begin{proof}
  By induction on the structure of $\phi$. The $\bot$ case holds because
  $\expl$ is maximal. All other cases are as usual.
\end{proof}

\begin{proposition}[Soundness \rocqdoc{Soundness.CK_soundness.html\#CK_Soundness}]\label{prop:soundness}
  If $\Gamma \vdash_{\log{CK}} \phi$, then $\Gamma \sement{CK} \phi$.
\end{proposition}
\begin{proof}
  By routine induction on the structure of a proof of
  $\Gamma \vdash \phi$. In particular, validity of (any instance of)
  the axiom $\bot \to p$ follows from maximality of $\expl$.
\end{proof}

  Next, we define a canonical model for the logic $\log{CK} \oplus \Ax$,
  where $\Ax$ is any set of axioms. This gives rise to a $\log{CK}$-model
  that validates precisely the consecutions derivable in
  $\log{CK} \oplus \Ax$. 
  We can use this to obtain completeness for a specific logic
  $\log{CK} \oplus \Ax'$ with respect to some class $\mathcal{F}$
  of $\log{CK}$-frames by showing that all frames in $\mathcal{F}$ validate
  the axioms in $\Ax'$, and that the $\log{CK}$-frame underlying the canonical
  model is in $\mathcal{F}$.
  In order to achieve the latter we sometimes have to modify the
  canonical model construction, as we will see in Section~\ref{sec:eight}.
  
  While canonical models are often based on \emph{theories} (Definition~\ref{def:theory}),
  we adapt Wijesekera's use of \emph{segments}~\cite{Wij90}.
  These are theories paired with a set of theories that intuitively denote their modal
  successors.
  This technique prevents distributivity of diamonds over disjunctions,
  corresponding to the \axref{ax:Cd} axiom.
  Because we have an exploding world,
  the theories we use to define our segments are allowed to contain $\bot$.

\begin{defn}\label{def:segment}
  A \emph{segment} (\rocqdoc{Complseg.general_seg_completeness.html\#segment})
  is a pair $(\Gamma, U)$ where $\Gamma$ is a prime theory
  and $U$ is a set of prime theories such that:
  \begin{enumerate}
    \item if $\Box\phi \in \Gamma$ then
          $\phi \in \Delta$ for all $\Delta \in U$;
    \item if $\Diamond\phi \in \Gamma$
          then $\phi \in \Delta$ for some $\Delta \in U$.
  \end{enumerate}
  
  Write $SEG$ for the set of all segments
  and define relations $\subsetsim$ (\rocqdoc{Complseg.general_seg_completeness.html\#cireach})
  and $R$ (\rocqdoc{Complseg.general_seg_completeness.html\#cmreach}) on $SEG$ by:
  \begin{align*}
    (\Gamma, U) \subsetsim (\Gamma', U') &\iff \Gamma \subseteq \Gamma' \\
    (\Gamma, U) R (\Gamma', U') &\iff \Gamma' \in U
  \end{align*}
\end{defn}

  Note that $\subsetsim$ defines a preorder on $SEG$.
  Furthermore, observe that $(\mathbf{L}, \{ \mathbf{L} \})$ is a segment,
  and conversely any segment of the form $(\mathbf{L}, U)$ must have
  $U = \{ \mathbf{L} \}$: $\Diamond\bot \in \mathbf{L}$ implies that
  $U$ is non-empty, and $\Box\bot \in \mathbf{L}$ implies that each of its
  elements is $\mathbf{L}$.
  Thus, setting $\expl = (\mathbf{L}, \{ \mathbf{L} \})$ (\rocqdoc{Complseg.general_seg_completeness.html\#cexpl})
  gives rise to a
  $\log{CK}$-frame $\mo{X}_{\log{CK\oplus\Ax}} := (SEG, \expl, \subsetsim, R)$ (\rocqdoc{Complseg.general_seg_completeness.html\#CF}).
  We can equip this frame with the valuation $V$ given by
  $V(p) = \{ (\Gamma, U) \in SEG \mid p \in \Gamma \}$
  for all $p \in \Prop$ (\rocqdoc{Complseg.general_seg_completeness.html\#cval}).
  Then we obtain the model $\mo{M}_{\log{CK}\oplus\Ax} = (\mo{X}_{\log{CK}\oplus\Ax}, V_{\Sigma})$
  (\rocqdoc{Complseg.general_seg_completeness.html\#CM}).

\begin{lemma}[\rocqdoc{Complseg.general_seg_completeness.html\#Diam_existence}]\label{lem:sigseg}
  Let $\Gamma$ be a prime theory such that $\Diamond\phi \notin \Gamma$.
  Then there exists a segment $(\Gamma, U)$ such that for all
  $\Delta \in U$ we have $\phi \notin \Delta$.
\end{lemma}
\begin{proof}
  For any $\Diamond\theta \in \Gamma$, we have
  $\Box^{-1}(\Gamma), \theta \noderiv{\Ax} \phi$,
  for otherwise we would get $\Gamma, \Diamond\theta \deriv{\Ax} \Diamond\phi$ 
  hence $\Diamond\phi \in \Gamma$.
  Now use the Lindenbaum lemma~\ref{lem:Lind} to find a prime theory $\Delta_{\theta}$
  containing $\Box^{-1}(\Gamma)$ and $\theta$ but not $\phi$.
  Then $(\Gamma, \{ \Delta_{\theta} \mid \Diamond\theta \in \Gamma \})$
  is a segment with the desired property.
\end{proof}

  In particular, the previous lemma implies that for each prime theory $\Gamma$
  we can construct a $\Sigma$-segment of the form $(\Gamma, U)$:
  if $\Diamond\bot \notin \Gamma$ we use Lemma~\ref{lem:sigseg}
  and if $\Diamond\bot \in \Gamma$ then we can take $U$ to be the set
  of prime theories containing $\Box^{-1}(\Gamma)$.

\begin{lemma}[Truth lemma \rocqdoc{Complseg.general_seg_completeness.html\#truth_lemma}]
  For any segment $(\Gamma, U)$ and formula $\phi \in \mathbf{L}$ we have
  $(\Gamma, U) \Vdash \phi$ iff $\phi \in \Gamma$.
\end{lemma}
\begin{proof}
  By induction on the structure of $\phi$.
  The cases for proposition letters and $\bot$ hold by construction.
  The inductive steps for meets, joins and implications are routine.
  
  If $\phi = \Diamond\psi$ then by construction
  $\Diamond\psi \in \Gamma$ implies $(\Gamma, U) \Vdash \Diamond\psi$.
  Conversely, if $\Diamond\psi \notin \Gamma$ then using
  Lemma~\ref{lem:sigseg} we can find a $\Sigma$-segment $(\Gamma, U')$ such
  that $\psi \notin \Delta$ for all $\Delta \in U'$.
  Since $(\Gamma, U) \subsetsim (\Gamma, U')$ we have
  $(\Gamma, U) \not\Vdash \Diamond\psi$ by persistence.
  
  Lastly, if $\phi = \Box\psi$ and $\Box\psi \in \Gamma$ then by construction
  we have $(\Gamma, U) \Vdash \Box\psi$.
  For the converse, suppose $\Box\psi \notin \Gamma$.
  Then $\Box^{-1}(\Gamma) \noderiv{\Ax} \psi$ 
  (for otherwise we would have $\Box\psi \in \Gamma$), so we can use the
  Lindenbaum lemma to find a prime theory $\Gamma_{\psi}$ containing
  $\Box^{-1}(\Gamma)$ but not $\psi$.
  Now we have that $(\Gamma, U \cup \{ \Gamma_{\psi} \})$ is a $\Sigma$-segment
  and $(\Gamma, U) \subsetsim (\Gamma, U \cup \{ \Gamma_{\psi} \})$,
  so that $\Gamma_{\psi}$ witnesses the fact that $(\Gamma, U) \not\Vdash \Box\psi$.
\end{proof}

\begin{theorem}[Strong completeness \rocqdoc{Complseg.general_seg_completeness.html\#Strong_Completeness}]\label{thm:completeness}
Let $\mathcal F$ be a class of frames such that $\mo{X}_{\log{CK\oplus\Ax}}\in\mathcal F$, 
and every $\mo{X}\in\mathcal F$ validates $\log{Ax}$.
Then, $\Gamma \sement{F} \phi$ entails $\Gamma \deriv{\Ax} \phi$ .
\end{theorem}
\begin{proof}
  We reason by contrapositive. Suppose $\Gamma \noderiv{\Ax} \phi$.
  Then we can find a prime theory $\Gamma'$ containing $\Gamma$ but not $\phi$,
  and extend $\Gamma$ to a segment of the form $(\Gamma', U)$~(\rocqdoc{Complseg.general_seg_completeness.html\#Lindenbaum_segment}). 
  The truth lemma implies $(\Gamma', U) \Vdash \chi$ for all
  $\chi \in \Gamma$ while $(\Gamma', U) \not\Vdash \phi$.
  Since $\mo{X}_{\log{CK\oplus\Ax}}\in\mathcal F$ by assumption,
  we find $\Gamma \nosement{F} \phi$.
\end{proof}

\begin{rem}\label{rem:CK-fmp}
  The canonical model construction can also be performed relative to a
  finite set $\Sigma$ of formulas. This gives rise to a finite
  canonical model, a truth lemma relative to $\Sigma$, and ultimately
  a finite model property. 
  Since this is beyond the scope
  of this paper we omit the details.
\end{rem}

As the class of frames $\mathcal{CK}$ vacuously validates all additional axioms of $\log{CK}$, i.e.~none, 
and the frame $\mo{X}_{\log{CK}}\in\mathcal{CK}$, we exploit the result above to obtain strong completeness
for $\log{CK}$.

\begin{theorem}[Strong completeness for $\log{CK}$ \rocqdoc{Complseg.CK_seg_completeness.html\#CK_Strong_Completeness}]\label{thm:ck-completeness}
  If $\Gamma \sement{CK} \phi$ then $\Gamma \deriv{\log{CK}} \phi$.
\end{theorem}

\begin{rem}
  Wijesekera~\cite{Wij90} uses a similar construction as above to obtain
  completeness for $\log{WK}$. Besides incorporating an inconsistent world,
  the main difference is that our canonical model is based on the set of all
  segments, while Wijesekera uses recursion to generate a model from a given
  segment.
  We also note that Wijesekera claims completeness with respect to partially
  ordered frames, but their canonical model construction only gives a preorder.
  The claim for partially ordered frames can be recovered via an unravelling
  construction akin to~\cite[Section~3.3]{ChaZak97}.
\end{rem}

%%%%%%%%%%%%%%%%%%%%%%%%%%%%%%%%%%%%%%%%%%%%%%%%%%%%%%%%%%%%%%%%%%%%%%%%%%%%%%%%
\section{Three axioms between $\log{CK}$ and $\log{IK}$}\label{sec:three-axioms}

  The logic $\log{IK}$ can be obtained by extending $\log{CK}$
  with $\axref{ax:Nd}$, $\axref{ax:Cd}$ and $\axref{ax:Idb}$.
  This section examines these three axioms individually.
  We give frame conditions that guarantee validity for each of them,
  and then refine these to frame correspondence conditions.
  While the latter conditions provide sound and strongly complete semantics for
  the extension of $\log{CK}$ with any combination of $\axref{ax:Nd}$,
  $\axref{ax:Cd}$ and $\axref{ax:Idb}$,
  we use the former when possible because of their greater simplicity.  
  In particular, we provide the birelational semantics for
  $\log{CK} \oplus \axref{ax:Nd} \oplus \axref{ax:Cd}$ called for
  by Das and Marin~\cite[Section 7]{DasMar23}.

\begin{defn}\label{def:frm-cond}
  Let $\mo{X} = (X, \expl, \leq, R)$ be a $\log{CK}$-frame.
  We identify three frame conditions~%
  (\rocqdoc{Kripke.correspondence.html\#suff_Nd_frame},% 
   \rocqdoc{Kripke.correspondence.html\#suff_Cd_frame},%
   \rocqdoc{Kripke.correspondence.html\#suff_Idb_frame}):
  \begin{myenumerate} \itemsep=3pt %[\qquad\qquad\qquad]
    \setlength{\itemindent}{2.5em}
    \myitem{\ref{ax:Nd}-suff} \label{eq:Nd-suff}
          $\forall x \; (xR\expl \text{ implies } x = \expl)$
          
    \myitem{\ref{ax:Cd}-suff} \label{eq:Cd-suff}
          $\begin{aligned}[t]
           \forall x \, \exists x' \; (&x \leq x'
            \text{ and } \forall y, z \; (
            \text{if } x \leq y
             \text{ and } x' R z \\
            &\text{then } \exists w \in X
             \text{ s.t. } yRw
             \text{ and } z \leq w))
          \end{aligned}$
          
    \myitem{\ref{ax:Idb}-suff} \label{eq:Idb-suff}
          $\begin{aligned}[t]
            &\forall x, y, z
              \text{ s.t } x R y \leq z \;(
              \exists u \in X
              \text{ s.t. } x \leq u R z \\
             &\hspace{-1em}\text{ and } \forall s \in X
              \text{ s.t. } u \leq s, \exists t \in X
              \text{ s.t. } sRt
              \text{ and } z \leq t)
           \end{aligned}$
  \end{myenumerate}
\end{defn}

\begin{proposition}
  Let $\mo{X} = (X, \expl, \leq, R)$ be a frame and
  $\sf{A} \in \{ \axref{ax:Nd}, \axref{ax:Cd}, \axref{ax:Idb} \}$.
  If $\mo{X}$ satisfies \textup{($\sf{A}$-suff)},
  then $\mo{X}$ validates $\sf{A}$.
\end{proposition}
\begin{proof}
  This result%
  ~(\rocqdoc{Kripke.correspondence.html\#sufficient_Nd},%
   \rocqdoc{Kripke.correspondence.html\#sufficient_Cd},%
   \rocqdoc{Kripke.correspondence.html\#sufficient_Idb})
  follows from the fact that each of the conditions above implies
  the correspondence condition of the axiom under consideration%
  ~(\rocqdoc{Kripke.correspondence.html\#suff_impl_Nd},%
   \rocqdoc{Kripke.correspondence.html\#suff_impl_Cd},%
   \rocqdoc{Kripke.correspondence.html\#suff_impl_Idb}),
  given below.
\end{proof}

\begin{rem}\label{rem:frm-cond}
  If we take $x = x'$ in~\eqref{eq:Cd-suff}, we obtain
  a stronger condition
  (\rocqdoc{Kripke.correspondence.html\#strong_Cd_frame}),
  \begin{myenumerate}
    \setlength{\itemindent}{3.3em}
    \myitem{\ref{ax:Cd}-strong} \label{eq:Cd-strong}
        $\begin{aligned}[t]
          \forall x, y, z (&\text{if } x \leq y \text{ and } xRz \\
            &\text{then } \exists w \; (yRw \text{ and } z \leq w)).
         \end{aligned}$
  \end{myenumerate}
  This is a standard frame condition for the semantics of
  $\log{IK}$~\cite{Fis81,PloSti86,Sim94}.
  It implies that we can ignore the universal quantifier in the interpretation
  of $\Diamond\phi$, looking only at modal successors of the current world.
  In presence of~\eqref{eq:Cd-strong}, condition~\eqref{eq:Idb-suff} is equivalent
  to~(\rocqdoc{Kripke.correspondence.html\#weak_Idb_frame})
  \begin{myenumerate}
    \setlength{\itemindent}{3.3em}
    \myitem{\ref{ax:Idb}-weak} \label{eq:Idb-suffCd}
        $\forall x, z, u$ (if $xRz \leq u$ then $\exists y$ ($x \leq y R u$)),
  \end{myenumerate}
  which is also used in the standard semantics of $\log{IK}$.
\end{rem}

  \begin{figure}[h!]
    \vspace{-1em}
    \begin{subfigure}{.15\textwidth}
    \centering
    \begin{tikzpicture}[scale=.6,xscale=.9]
      %% nodes
      \node (x) at (0,0) {$x$};
      \node (e) at (2.5,0) {$\expl$};
      \node (phantom) at (0,-1) {};
      %% edges
      \draw[-latex, dashed] (x) \rto{above} (e);
      %% cross
      \draw[very thick,red] (.95,.3) -- (1.55,-.3);
      \draw[very thick,red] (.95,-.3) -- (1.55,.3);
    \end{tikzpicture}
    \caption{\eqref{eq:Nd-suff}}
    \end{subfigure}
    \begin{subfigure}{.15\textwidth}
    \centering
    \hspace{-1em}
    \begin{tikzpicture}[scale=.6,xscale=.9]
      %% nodes
      \node (x) at (0,-.4) {$x$};
      \node (xp) at (0,1.5) {$x'$};
      \node (y) at (0,3) {$y$};
      \node (z) at (2.5,1.5) {$z$};
      \node (w) at (2.5,3) {$w$};
      %% edges
      \draw[-latex, bend left=35] (x) \ito{left} (y);
      \draw[-latex, dashed ] (x) \ito{right} (xp);
      \draw[-latex, dashed] (z) \ito{right} (w);
      \draw[-latex] (xp) \rto{below} (z);
      \draw[-latex, dashed] (y) \rto{above} (w);
    \end{tikzpicture}
    \caption{\eqref{eq:Cd-suff}}
    \end{subfigure}
    \begin{subfigure}{.15\textwidth}
    \centering
    \begin{tikzpicture}[scale=.6,xscale=.9]
      %% nodes
      \node (x) at (0,0) {$x$};
      \node (y) at (2.5,0) {$y$};
      \node (z) at (2.5,1.5) {$z$};
      \node (u) at (0,1.5) {$u$};
      \node (s) at (0,3) {$s$};
      \node (t) at (2.5,3) {$t$};
      %\node (e) at (2.6,6.4) {$\expl$};
      %% edges
      \draw[-latex] (x) \rto{below} (y);
      \draw[-latex, bend right=5] (y) to (z);
      \draw[-latex, dashed, bend left=5] (x) \ito{left} (u);
      \draw[-latex, dashed] (u) \rto{below} (z);
      \draw[-latex, bend left=5] (u) to (s);
      \draw[-latex, dashed, bend right=5] (z) to (t);
      \draw[-latex, dashed] (s) \rto{above} (t);
    \end{tikzpicture}
    \caption{\eqref{eq:Idb-suff}}
    \label{fig:Idb-suff}
    \end{subfigure}
    \begin{subfigure}{.15\textwidth}
    \centering
    \begin{tikzpicture}[scale=.6,xscale=.9]
      %% nodes
      \node (x) at (0,0) {$x$};
      \node (y) at (0,2) {$y$};
      \node (z) at (2.5,0) {$z$};
      \node (w) at (2.5,2) {$w$};
      %% edges
      \draw[-latex] (x) \ito{left} (y);
      \draw[-latex, dashed] (z) \ito{right} (w);
      \draw[-latex] (x) \rto{below} (z);
      \draw[-latex, dashed] (y) \rto{above} (w);
    \end{tikzpicture}
    \caption{\eqref{eq:Cd-strong}}
    \end{subfigure}
    \begin{subfigure}{.15\textwidth}
    \centering
    \begin{tikzpicture}[scale=.6,xscale=.9]
      %% nodes
      \node (x) at (0,0) {$x$};
      \node (y) at (0,2) {$y$};
      \node (z) at (2.5,0) {$z$};
      \node (w) at (2.5,2) {$u$};
      \node (phantom) at (0,3) {}; % for extra vertical space between the subfigures
      %% edges
      \draw[-latex, dashed] (x) to (y);
      \draw[-latex] (z) to (w);
      \draw[-latex] (x) \rto{below} (z);
      \draw[-latex, dashed] (y) \rto{above} (w);
    \end{tikzpicture}
    \caption{\eqref{eq:Idb-suffCd}}
    \end{subfigure}
  \caption{Sufficient conditions for validity of $\axref{ax:Nd}$,
           $\axref{ax:Cd}$ and $\axref{ax:Idb}$.
           Unlabelled arrows denote the intuitionistic relation.
           Solid arrows indicate universally quantified relations,
           while the dashed ones indicate existential ones.}
  \label{fig:frm-cond}
  \end{figure}
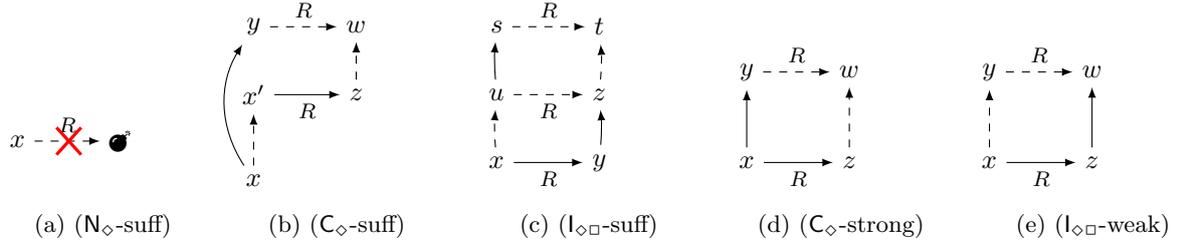

  The five frame conditions introduced in Definition~\ref{def:frm-cond}
  and Remark~\ref{rem:frm-cond} are depicted in Figure~\ref{fig:frm-cond}.
  While these are sufficient to ensure
  validity of certain axioms, none of them are necessary.
  The next examples illustrate this for~\eqref{eq:Nd-suff}
  and~\eqref{eq:Idb-suff}.

\begin{example}\label{exm:Nd-suff}
  Consider the frame depicted in Figure~\ref{fig:exm-Nd-suff}.
  This validates $\Diamond\bot \to \bot$, because $\expl$ is the only
  world that satisfies $\Diamond\bot$. But it does not satisfy
  \eqref{eq:Nd-suff}, because $xR\expl$.
\end{example}

\begin{example}\label{exm:Idb-suff}
  Consider the frame given in Figure~\ref{fig:exm-Idb-suff}.
  This satisfies neither~\eqref{eq:Idb-suff} nor~\eqref{eq:Idb-suffCd}.
  However, the frame does validate $\axref{ax:Idb}$.
  To see this, we show that every world that satisfies $\Diamond p \to \Box q$
          also satisfies $\Box(p \to q)$.
          For $u, v, w$ this follows immediately from their lack of
          modal successors, which implies that they trivially satisfy $\Box(p \to q)$.
          For $y$ it follows from the fact that
          $y \Vdash \Diamond p \to \Box q$ implies that either
          $w \not\Vdash p$ or $w \Vdash q$, so that $w \Vdash p \to q$
          whence $y \Vdash \Box(p \to q)$.
          Lastly, suppose $x \Vdash \Diamond p \to \Box q$.
          If none of $u, v, w$ satisfy $p$ then they all satisfy $p \to q$,
          and hence $x \Vdash \Box(p \to q)$.
          If any of $u, v, w$ satisfy $p$ then so does $w$,
          which implies $x \Vdash \Diamond p$.
          Then we have $x \Vdash \Box q$, so that $u \Vdash q$ and hence
          $u, v, w \Vdash q$.
          Thus $u, v, w \Vdash p \to q$, which again
          implies $x \Vdash \Box(p \to q)$.
\end{example}

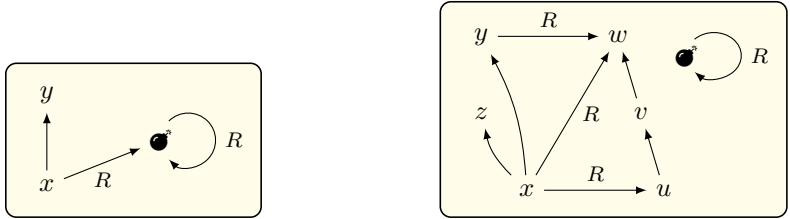
\begin{figure}[h!]
  \vspace{-1em}
    \begin{subfigure}{.23\textwidth}
      \centering
      \begin{tikzpicture}[scale=.6]
      %% Frame 1
      \draw[rounded corners, fill=blue!9, semithick] (-.9,2.7) rectangle (4.7,-.7);
        %% nodes
        \node (x1) at (0,0) {$x$};
        \node (y1) at (0,2) {$y$};
        \node (e1) at (2.5,1) {$\expl$};
        %% edges
        \draw[-latex] (x1) \rto{below} (e1);
        \draw[-latex] (x1) \ito{left} (y1);
        \draw[-latex] ([shift=(135:.6)]3.1,1) arc (135:-135:.6)
                      node[right,pos=.5]{\footnotesize{$R$}};
      \end{tikzpicture}
      \caption{Frame from Exm.~\ref{exm:Nd-suff}.}
      \label{fig:exm-Nd-suff}
    \end{subfigure}
    \begin{subfigure}{.25\textwidth}
      \begin{tikzpicture}[scale=.6,yscale=.85]
      \draw[rounded corners, fill=blue!9, semithick] (14.1,4.9) rectangle (21.7,-.7);
        %% nodes
        \node (x3) at (16,0) {$x$};
        \node (y3) at (15,4) {$y$};
        \node (u3) at (19,0) {$u$};
        \node (v3) at (18.5,2) {$v$};
        \node (w3) at (18,4) {$w$};
        \node (e3) at (19.5,3.5) {$\expl$};
        %% edges
        \draw[-latex,bend right=10] (x3) \ito{right,pos=.7} (y3);
        \draw[-latex] (x3) \rto{above} (u3);
        \draw[-latex] (x3) \rto{right} (w3);
        \draw[-latex] (u3) \ito{right} (v3);
        \draw[-latex] (v3) \ito{right,pos=.3} (w3);
        \draw[-latex] (y3) \rto{above} (w3);
        \draw[-latex] ([shift=(135:.6)]20.1,3.5) arc (135:-135:.6) node[right,pos=.5]{\footnotesize{$R$}};
      \end{tikzpicture}
      \caption{Frame from Example~\ref{exm:Idb-suff}.}
      \label{fig:exm-Idb-suff}
    \end{subfigure}
  \caption{Frames witnessing non-necessity of some of the sufficient
           conditions from Definition~\ref{def:frm-cond} and Remark~\ref{rem:frm-cond}.}
\end{figure}

  We now give exact correspondence conditions for each of
  the axioms $\axref{ax:Nd}, \axref{ax:Cd}$ and $\axref{ax:Idb}$.
  Two of the three correspondence conditions are depicted in
  Figure~\ref{fig:corr} below.
  
\begin{proposition}[\rocqdoc{Kripke.correspondence.html\#correspond_Nd}]
  A $\log{CK}$-frame $\mo{X} = (X, \expl, \leq, R)$ validates $\axref{ax:Nd}$
  if and only if it satisfies:
  \begin{myenumerate}
    \setlength{\itemindent}{3.3em}
    \myitem{\ref{ax:Nd}-corr} \label{eq:Nd-corr}
      $\forall x \; (\text{if } yR\expl \text{ for all } y \geq x,
       \text{ then } x = \expl)$
  \end{myenumerate}
\end{proposition}
\begin{proof}
  The correspondence condition implies validity of $\axref{ax:Nd}$ by
  definition. Conversely, if the correspondence condition does not hold
  then this must be witnessed by a world that satisfies $\Diamond\bot$
  but not $\bot$.
\end{proof}

  To simplify the statement of the correspondence condition for $\axref{ax:Cd}$
  we use the following notation: if $(X, \leq)$ is a preorder, $R$ a binary relation on $X$,
  $x\in X$ and $a \subseteq X$, then
  ${\downarrow}a := \{ x \in X \mid x \leq y \text{ for some } y \in a \}$ denotes
  the downset generated by $a$,
  $R[x]:=\{ y \in X \mid xRy \}$ and $R^{-1}(x) = \{ y \in X \mid yRx \}$.

\begin{proposition}[\rocqdoc{Kripke.correspondence.html\#correspond_Cd}]
  A frame $\mo{X} = (X, \expl, \leq, R)$ validates $\axref{ax:Cd}$
  if and only if it satisfies:
  \begin{myenumerate}
    \setlength{\itemindent}{2.3em}
    \myitem{\ref{ax:Cd}-corr} \label{eq:Cd-corr}
      $\begin{aligned}[t]
         \forall x, y, z \; (
           &\text{if } y, z \notin R^{-1}(\expl)
            \text{ and } x \leq y
            \text{ and } x \leq z \\
           &\text{ then } \exists w ( x \leq w 
            \text{ and } R[w] \subseteq {\downarrow}R[y] \\
           &\phantom{\text{ then } \exists w ( x \leq w}
            \text{ and } R[w] \subseteq {\downarrow}R[z]))
      \end{aligned}$
  \end{myenumerate}
\end{proposition}
\begin{proof}
  Suppose $\mo{X}$ satisfies~\eqref{eq:Cd-corr}.
  Let $x$ be a world, $V$ any valuation and suppose
  $x \not\Vdash \Diamond p \vee \Diamond q$.
  Then $x \not\Vdash \Diamond p$ and $x \not\Vdash \Diamond q$, so there exist
  $y \geq x$ and $z \geq x$
  such that $R[y] \cap V(p) = \emptyset$ and $R[z] \cap V(q) = \emptyset$.
  We must have $y, z \notin R^{-1}(\expl)$, so
  by~\eqref{eq:Cd-corr} we get some $w \geq x$ such that
  every $R$-successor of $w$ lies below an $R$-successor of $y$
  and below an $R$-successor of $z$. This implies that $R[w] \cap (V(p) \cup V(q)) = \emptyset$, so that
  $w$ witnesses $x \not\Vdash \Diamond(p \vee q)$.
  Since this holds for every $x \in X$ and every valuation,
  it follows that $\mo{X} \Vdash \Diamond(p \vee q) \to \Diamond p \vee \Diamond q$.
  
  For the converse, suppose that~\eqref{eq:Cd-corr} does not hold.
  Then we can find $x, y, z$ satisfying $y, z \notin R^{-1}(\expl)$ and
  $x \leq y$ and $x \leq z$,
  and such that there is no $w \geq x$ such that every $R$-successor of $w$
  lies below $R$-successors of $y$ and $z$.
  Taking $V(p) = (X \setminus {\downarrow}R[y]) \cup \{ \expl \}$
  and $V(q) = (X \setminus {\downarrow}R[z]) \cup \{ \expl \}$
  then results in a model where $x$ satisfies $\Diamond(p \vee q)$
  but not $\Diamond p$ or $\Diamond q$, whence
  $x \not\Vdash \Diamond(p \vee q) \to \Diamond p \vee \Diamond q$.
\end{proof}

\begin{proposition}[\rocqdoc{Kripke.correspondence.html\#correspond_Idb}]\label{prop:Idb-corr}
  A frame $\mo{X} = (X, \expl, \leq, R)$ validates $\axref{ax:Idb}$
  if and only if it satisfies:
  \begin{myenumerate}
    \setlength{\itemindent}{2.7em}
    \myitem{\ref{ax:Idb}-corr} \label{eq:Idb-corr}
      $\begin{aligned}[t]
       \forall x, y, z (
        &\text{if } x R y \leq z \neq \expl \\
        &\text{then } \exists u, w \; (
          x \leq u R w \leq z \\
          &\text{ and } \forall s \;
          (\text{if } u \leq s 
           \text{ then } sR \expl \\
          &\phantom{\text{ and } \forall s \;}\text{ or } \exists t (sRt \text{ and } z \leq t))))
      \end{aligned}$
  \end{myenumerate}
\end{proposition}
\begin{proof}
  Suppose $\mo{X}$ satisfies~\eqref{eq:Idb-corr}.
  Let $V$ be any valuation and suppose $x' \not\Vdash \Box(p \to q)$.
  Then there exist worlds $x$, $y$ and $z$ such that $x' \leq xRy \leq z$ and
  $z \Vdash p$ and $z \not\Vdash q$.
  Since $z \not\Vdash q$ we must have $z \neq \expl$, so we can 
  find $u, w$ with the properties mentioned in~\eqref{eq:Idb-corr}.
  Then $u \Vdash \Diamond p$ because each intuitionistic successor $s$ of $u$
  can modally see $\expl$ or some successor of $z$, both of which satisfy $p$.
  But we have $u \not\Vdash \Box q$, because $u \leq u R w$
  and $w \leq z$, hence $w \not\Vdash q$.
  Since $x' \leq x \leq u$, we conclude $x' \not\Vdash \Diamond p \to \Box q$.
  
  For the converse, suppose that the frame validates
  $(\Diamond p \to \Box q) \to \Box(p \to q)$ and let $x, y, z \in X$ be
  such that $x R y \leq z$ and $z \neq \expl$.
  Since the formula is valid, it holds for all valuations and we
  can exploit its truth under any valuation that suits us.
  Define $V(p) = ({\uparrow}z) \cup \{ \expl \}$ and
  $V(q) = (X \setminus {\downarrow}z) \cup \{ \expl \}$.
  Then $z \Vdash p$ and $z \not\Vdash q$ (because $z \neq \expl$),
  and therefore $y \not\Vdash p \to q$,
  so that $x \not\Vdash \Box(p \to q)$.
  But this means that we must have $x \not\Vdash \Diamond p \to \Box q$.
  So there exists some successor $v$ of $x$ such that $v \Vdash \Diamond p$
  while $v \not\Vdash \Box q$.
  The latter implies that there exist worlds $u, w$ such that
  $v \leq u R w$ and $w \not\Vdash q$. By definition of $V(q)$ this means
  $w \leq z$. Furthermore, $v \Vdash \Diamond p$ implies $u \Vdash \Diamond p$,
  so each intuitionistic successor $s$ of $u$ has a modal successor that
  satisfies $p$. By definition this means that either $sR\expl$ or there
  is some $t$ such that $sRt$ and $z \leq t$.
  Thus we have found $u, w \in X$ with the desired properties, hence
  the frame satisfies~\eqref{eq:Idb-corr}.
\end{proof}

  \begin{figure}[h]
    \centering
    \begin{subfigure}{.21\textwidth}
    \begin{tikzpicture}[xscale=.75]
      %% Notes
      \node (x) at (0,0) {$x$};
      \node (y) at (-.6,2.2) {$y$};
      \node (z) at (.8,1.5) {$z$};
      \node (w) at (0,.9) {$w$};
      \node (wp) at (3,1) {$w'$};
      \node (yp) at (2.3,2.2) {$y'$};
      \node (zp) at (4.1,1.6) {$z'$};
      %% Edges
      \draw[-latex, dashed] (x) to (w);
      \draw[-latex, bend left=20] (x) to (y);
      \draw[-latex, bend right=20] (x) to (z);
      \draw[-latex] (w) to node[below]{\footnotesize{$R$}} (wp);
      \draw[-latex, dashed, bend left=15] (wp) to (yp);
      \draw[-latex, dashed, bend right=25] (wp) to (zp);
      \draw[-latex, dashed] (y) to node[above]{\footnotesize{$R$}} (yp);
      \draw[-latex, dashed] (z) \rto{below,pos=.35} (zp);
    \end{tikzpicture}
    \caption{\eqref{eq:Cd-corr}}
    \end{subfigure}
    \qquad
    \begin{subfigure}{.2\textwidth}
    \vspace{-1em}
    \begin{tikzpicture}[scale=.55, yscale=.8,xscale=1]
      %% nodes
      \node (x) at (0,0) {$x$};
      \node (y) at (2.5,0) {$y$};
      \node (z) at (3.5,3.25) {$z$};
      \node (znb) at (4.35,3.25) {$\neq \expl$};
      \node (u) at (-1,2.5) {$u$};
      \node (w) at (2.5,1.25) {$w$};
      \node (s) at (0,5.7) {$s$};
      \node (t) at (2.5,5) {$t$};
      \node (e) at (2.6,6.4) {$\expl$};
      %% edges
      \draw[-latex] (x) \rto{below} (y);
      \draw[-latex, bend right=15] (y) to (z);
      \draw[-latex, dashed, bend left=15] (x) to (u);
      \draw[-latex, dashed, bend right=10] (u) \rto{above,pos=.5} (w);
      \draw[-latex, dashed, bend right=5] (w) to (z);
      \draw[-latex, bend left=15] (u) to (s);
      \draw[-latex, dashed, bend right=15] (z) to (t);
      \draw[-, dashed] (s) to (1,5.7);
      \draw[-latex, dashed, bend right=11]  (1.05,5.7) \rto{above,pos=0} (e);
      \draw[-latex, dashed, bend left=11] (1.05,5.7) to (t);
      \node at (2.5,5.7) {\footnotesize{or}};
    \end{tikzpicture}
    \caption{\eqref{eq:Idb-corr}}
    \label{fig:Idb-corr}
    \end{subfigure}
    \caption{Correspondence conditions for validity of $\axref{ax:Cd}$ and $\axref{ax:Idb}$.
           The solid arrows indicate universally quantified arrows,
           while the dashed ones indicate existential ones.
           The modal relation is labelled $R$ and the intuitionistic relation
           is unlabelled.}
    \label{fig:corr}
  \end{figure}
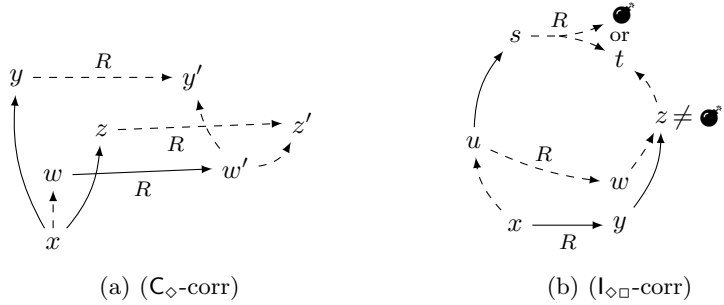

%%%%%%%%%%%%%%%%%%%%%%%%%%%%%%%%%%%%%%%%%%%%%%%%%%%%%%%%%%%%%%%%%%%%%%%%%%%%%%%%
\section{Eight completeness results}\label{sec:eight}

  The correspondence results for the axioms
  $\axref{ax:Nd}, \axref{ax:Cd}$ and $\axref{ax:Idb}$ give rise to
  sound semantics for extensions of $\log{CK}$ with any combination of them.
  Next, we prove corresponding completeness results.
  Where possible, we show that the canonical model satisfies the
  sufficient frame conditions, to simplify our reasoning about conservativity
  in Section~\ref{sec:diamond-free}.

\begin{theorem}\label{thm:strong_compl_for_Ax}
  Let $\Ax \subseteq \{ \axref{ax:Nd}, \axref{ax:Cd}, \axref{ax:Idb} \}$.
  Then the logic $\log{CK} \oplus \Ax$ is sound and
  strongly complete with respect to the class of frames satisfying
  \textup{($\sf{A}$-suff)} for each $\sf{A} \in \Ax$,
  except for $\log{CK} \oplus \axref{ax:Nd}$ and $\log{CK} \oplus \axref{ax:Nd} \oplus \axref{ax:Cd}$, for which
  \eqref{eq:Nd-corr} replaces \eqref{eq:Nd-suff}.
\end{theorem}

  The completeness part of the proof splits into four cases, depending on
  whether or not $\axref{ax:Cd}$ and $\axref{ax:Idb}$ are in $\Ax$.
  Each uses a slightly different canonical model construction.
  Intuitively, we ``prune'' the canonical model construction for $\log{CK}$
  from Definition~\ref{def:segment} (i.e.~we leave out certain segments)
  to ensure satisfaction of relevant sufficient and correspondence conditions.
  We begin by adapting the canonical model to accommodate for the case where
  $\axref{ax:Cd} \in \Ax$ and $\axref{ax:Idb} \notin \Ax$.

\begin{defn}[\rocqdoc{ComplsegAB.general_segAB_completeness.html\#ABsegment}]
  Let $\Gamma$ and $\Delta$ be prime theories. 
  \begin{enumerate}
    \item The \emph{A-segment} of $\Gamma$ is the segment
          $(\Gamma, A)$ where $\Delta \in A$ if and only if for all $\phi \in \mathbf{L}$:
          $\Box\phi \in \Gamma$ implies $\phi \in \Delta$,
          and $\phi \in \Delta$ implies $\Diamond\phi \in \Gamma$.
    \item The \emph{B-segment} of $\Gamma$ is the segment
          $(\Gamma, B)$ where $\Delta \in B$ if and only if
          $\{ \phi \mid \Box\phi \in \Gamma \} \subseteq \Delta$.
  \end{enumerate}
\end{defn}

  Note that we always have $\expl \in B$.
  Besides, $\expl \in A$ if $\Diamond\bot \in \Gamma$.
  We now verify that $(\Gamma, A)$ and $(\Gamma, B)$ are indeed segments.

\begin{lemma}
  Let $\Gamma$ be a prime theory with A- and B-segment $(\Gamma, A)$ and $(\Gamma, B)$.
  Then we have:
  \begin{enumerate}
    \item If $\Box\phi \in \Gamma$ then $\phi \in \Delta$ for all $\Delta \in A \cup B$%
             ~$(\rocqdoc{ComplsegAB.general_segAB_completeness.html\#boxreflect})$.
    \item If $\Diamond\phi \in \Gamma$ then $\phi \in \Delta$ for some $\Delta \in A$,
          and $\phi \in \Theta$ for some $\Theta \in B$%
          ~$(\rocqdoc{ComplsegAB.general_segAB_completeness.html\#diamwitness})$.
  \end{enumerate}  
\end{lemma}
\begin{proof}
  Item 1) holds by construction, and for B-segments 2) is witnessed by
  $\Theta = \expl \in B$, so we are left to prove 2) for A-segments.
  To this end, suppose $\Diamond\phi \in \Gamma$.
  If $\Diamond\bot \in \Gamma$, then we can also use $\Delta = \expl\in A$ as witness.
  Else, we claim that 
  \begin{equation}\label{eq:Cd-truth}
    \Box^{-1}(\Gamma) \cup \{ \phi \}
      \not\vdash \Diamond^{-1}(\Gamma^c).
  \end{equation}
  (Here $\Gamma^c = \lan{L} \setminus \Gamma$,
  so $\Diamond^{-1}(\Gamma^c) = \{ \theta \mid \Diamond\theta \notin \Gamma \}$.)
  Then $\Diamond^{-1}(\Gamma^c)$ is not empty as $\Diamond\bot \not\in \Gamma$.
  If~\eqref{eq:Cd-truth} is false, then there are
  $\psi_1, \ldots, \psi_n \in \Box^{-1}(\Gamma)$
  and $\theta_1, \ldots, \theta_m \in \Diamond^{-1}(\Gamma^c)$ such that
  $\psi_1 \wedge \cdots \wedge \psi_n \wedge \phi \vdash \theta_1 \vee \cdots \vee \theta_m$.
  Since $\Box$ distributes over meets and $\Diamond$ distributes over joins%
  ~(\rocqdoc{GHC.properties.html\#Diam_distrib_list_disj}),
  we have $\psi := \psi_1 \wedge \cdots \psi_n \in \Box^{-1}(\Gamma)$
  and $\theta := \theta_1 \vee \cdots \vee \theta_m \in \Diamond^{-1}(\Gamma^c)$.
  Then $\psi \wedge \phi \vdash \theta$, so by~\axref{ax:Kd},
  $\Box\psi \wedge \Diamond\phi \vdash \Diamond\theta$.
  By assumption $\Box\psi, \Diamond\phi \in \Gamma$, hence $\Diamond\theta \in \Gamma$,
  a contradiction.
  Now the Lindenbaum lemma yields a prime theory $\Delta$ containing
  $\Box^{-1}(\Gamma) \cup \{ \phi \}$ and disjoint from $\Diamond^{-1}(\Gamma^c)$.
  Therefore $\phi \in \Delta$ and $\Delta \in A$, as desired.
\end{proof}

  We construct the canonical frame $\mo{X}_{AB}$ and model $\mo{M}_{AB}$%
  ~(\rocqdoc{ComplsegAB.general_segAB_completeness.html\#CM})
  as in Definition~\ref{def:segment}, except that we restrict our worlds
  to A- and B-segments.
  Then we have the following truth lemma.
  
\begin{lemma}[Truth lemma \rocqdoc{ComplsegAB.general_segAB_completeness.html\#truth_lemma}]
  Let $(\Gamma, U)$ be an A- or B-segment. Then for all $\phi \in \mathbf{L}$ we have
  $(\Gamma, U) \Vdash \phi$ iff $\phi \in \Gamma$.
\end{lemma}
\begin{proof}
  By induction on the structure of $\phi$. 
  We demonstrate the case $\phi = \Box\psi$.
  If $\Box\psi \in \Gamma$ then $\Gamma \Vdash \Box\psi$ by construction.
  If $\Box\psi \notin \Gamma$ then $\Box^{-1}(\Gamma) \not\vdash \psi$,
  so we can use the Lindenbaum lemma to find a prime theory $\Delta$
  containing $\Box^{-1}(\Gamma)$ but not~$\psi$.
  By definition $\Delta$ is in the tail of the B-segment of $\Gamma$.
  Extending $\Delta$ to a segment $(\Delta, V)$ (using either the A- or B-segment),
  we find $(\Gamma, U) \subsetsim (\Gamma, B) R (\Delta, V)$,
  and by the induction hypothesis $(\Delta, V) \not\Vdash \psi$.
  Therefore $(\Gamma, U) \not\Vdash \Box\psi$.
\end{proof}
  
  Finally we check that restricting to A- and B-segments gives rise to
  a frame that satisfies~\eqref{eq:Cd-suff}.
  
\begin{lemma}[\rocqdoc{ComplsegAB.general_segAB_completeness.html\#CF_suff_Cd}]\label{lem:CF_suff_Cd}
  The frame $\mo{X}_{AB}$ satisfies~\eqref{eq:Cd-suff}.
\end{lemma}
\begin{proof}
  Given $(\Gamma, U)$, we claim that the A-segment $(\Gamma, A)$ of $\Gamma$
  witnesses satisfaction of~\eqref{eq:Cd-suff}.
  Suppose $(\Gamma, U) \subsetsim (\Gamma', U')$ and
  $(\Gamma, A) R (\Delta, V)$ (so $\Delta \in A$).
  If $\expl \in U'$ then $(\expl, A)$ is a segment such that
  $(\Gamma', U') R (\expl, A)$ and $(\Delta, V) \subsetsim (\expl, A)$,
  so~\eqref{eq:Cd-suff} holds.
  So assume $\expl \notin U'$. Then $(\Gamma', U')$ is an A-segment.
  We construct a $\Delta' \in U'$ such that $\Delta \subseteq \Delta'$.
  Then e.g.~the A-segment $(\Delta', A)$ witnesses truth of~\eqref{eq:Cd-suff}.
  
  First, note that $\{ \psi \mid \Diamond\psi \notin \Gamma' \} \neq \emptyset$,
  else $\Diamond\bot\in\Gamma'$ would entail the presence of $\expl$ in $U'$.
  Second, we claim that 
  \begin{equation}\label{eq:corr-Cd2}
    \Box^{-1}(\Gamma') \cup \Delta
      \not\vdash \{ \psi \mid \Diamond\psi \notin \Gamma' \}.
  \end{equation}
  If this were not the case, then we can find $\Box\phi \in \Gamma'$ and $\delta \in \Delta$ and $\Diamond\psi \notin \Gamma'$
  such that $\phi \wedge \delta \vdash \psi$, hence $\Box\phi \wedge \Diamond\delta \vdash \Diamond\psi$.
  By definition of A-segments, $\delta \in \Delta$ implies $\Diamond\delta \in \Gamma$,
  so $\Diamond\delta \in \Gamma'$. But $\Box\phi \in \Gamma'$ by
  construction, hence $\Diamond\psi \in \Gamma'$, a contradiction.
  So~\eqref{eq:corr-Cd2} holds, and we can use the Lindenbaum lemma to find
  a prime theory $\Delta'$ containing $\Delta$ and $\phi$ for each
  $\Box\phi \in \Gamma'$, while avoiding $\psi$ for each $\Diamond\psi \notin \Gamma'$.
  We claim that $\Diamond\phi\in\Gamma'$ for each $\phi\in\Delta'$.
  If not, then there exists a $\phi \in \Delta'$ such that 
  $\Diamond\phi \notin \Gamma'$, contradicting the fact that
  $\Delta'$ avoids all $\psi$ such that $\Diamond\psi \notin \Gamma'$.
  This entails $\Delta \subseteq \Delta'$ and $\Delta' \in U'$.
\end{proof}

  Combining the previous lemmas yields:

\begin{lemma}[\rocqdoc{ComplsegAB.CK_Cd_segAB_completeness.html\#suff_CKCd_Strong_Completeness}]
  $\log{CK} \oplus \axref{ax:Cd}$ is sound and strongly complete
  with respect to the class of frames satisfying~\eqref{eq:Cd-suff}.
\end{lemma}

  Next, we consider the case where $\axref{ax:Idb} \in \Ax$ and
  $\axref{ax:Cd} \notin \Ax$.

\begin{defn}[\rocqdoc{ComplsegP.general_segP_completeness.html\#Psegment}]\label{def:p-seg}
  A \emph{P-segment} (or \emph{purposeful segment}) of $\Gamma$ is defined as
  follows:
  \begin{itemize}
    \item If $\bot \in \Gamma$, then the only P-segment based on $\Gamma$ is
          $(\Gamma, \{ \Gamma \})$ (i.e.~$(\expl, \{ \expl \})$);
    \item If $\bot \notin \Gamma$ but $\Diamond\bot \in \Gamma$, then
          the only P-segment based on $\Gamma$ is $(\Gamma, U)$, where
          $U = \{ \Delta \mid \Box^{-1}(\Gamma) \subseteq \Delta \}$;
    \item If $\bot,\Diamond\bot \notin \Gamma$, then a segment $(\Gamma, U)$
          is a P-segment if there exists a $\Diamond\pi \notin \Gamma$ such that
          \begin{equation*}
            U = \{ \Delta \mid \Box^{-1}(\Gamma) \subseteq \Delta
                   \text{ and } \pi \notin \Delta \} =: U_{\Gamma,\pi}.
          \end{equation*}
  \end{itemize}
\end{defn}

  It is clear that the two cases with $\Diamond\bot \in \Gamma$ give rise
  to segments. For the case with $\Diamond\bot \notin \Gamma$, we have:
  
\begin{lemma}[\rocqdoc{ComplsegP.general_segP_completeness.html\#Lindenbaum_Psegment_Diam}]\label{lem:seg-Idb}
  Let $\Gamma$ be a prime theory such that $\Diamond\pi \notin \Gamma$.
  Let $U_{\Gamma,\pi} := \{ \Delta \mid \Box^{-1}(\Gamma) \subseteq \Delta \text{ and } \pi \notin \Delta \}$.
  Then $(\Gamma, U_{\Gamma,\pi})$ is a segment.
\end{lemma}
\begin{proof}
  By definition $\Box\phi \in \Gamma$ implies $\phi \in \Delta$.
  Let $\Diamond\phi \in \Gamma$.
  Then $\Box^{-1}(\Gamma) \cup \{ \phi \} \not\vdash \pi$,
  for if this were not the case then $\Gamma, \Diamond\phi \vdash \Diamond\pi$,
  contradicting $\Diamond\pi \notin \Gamma$.
  So the Lindenbaum lemma yields a prime theory $\Delta$
  containing $\Box^{-1}(\Gamma)$ and $\phi$ but not $\pi$,
  whence $\Delta \in U$ and $\phi \in \Delta$ as desired.
\end{proof}

  As a consequence of this lemma, for any prime theory
  $\Gamma$ there exists a P-segment headed by $\Gamma$.
  Construct the canonical frame $\mo{X}_P$ and model $\mo{M}_P$
  as in Definition~\ref{def:segment}, but restricting the set of segments
  to P-segments~(\rocqdoc{ComplsegP.general_segP_completeness.html\#CM}).

\begin{lemma}[Truth lemma \rocqdoc{ComplsegP.general_segP_completeness.html\#truth_lemma}]
  Let $(\Gamma, U)$ be a P-segment in $\mo{M}_P$.
  Then for all $\phi \in \mathbf{L}$ we have
  $(\Gamma, U) \Vdash \phi$ iff $\phi \in \Gamma$.
\end{lemma}
\begin{proof}
  By induction on the structure of $\phi$. The only interesting cases
  are for $\phi = \Box\psi$ and $\phi = \Diamond\psi$.
  
  If $\Box\psi \in \Gamma$ then $\Gamma \Vdash \Box\psi$ by construction.
  If $\Box\psi \notin \Gamma$ then $\Box^{-1}(\Gamma) \not\vdash \psi$.
  Use the Lindenbaum lemma to construct a prime theory $\Delta$
  containing $\Box^{-1}(\Gamma)$ but not $\psi$, and extend this to a P-segment
  $(\Delta, V)$.
  Now we consider two cases:
  if $\Diamond\bot \in \Gamma$ then $\Delta$ is in the tail of the (unique)
  P-segment headed by $\Gamma$ (by definition).
  Then $(\Gamma, U)R(\Delta, V)$ and $(\Delta, V) \not\Vdash \psi$,
  hence $(\Gamma, U) \not\Vdash \Box\psi$.
  If $\Diamond\bot \notin \Gamma$ then $\Delta \in U_{\Gamma,\bot}$,
  so $(\Gamma, U_{\Gamma,\bot})$ is a segment above $(\Gamma, U)$
  witnessing $(\Gamma, U) \not\Vdash \Box\psi$.
  
  If $\Diamond\psi \in \Gamma$ then $(\Gamma, U) \Vdash \Diamond\psi$ by
  the definition of segments. If $\Diamond\psi \notin \Gamma$
  then Lemma~\ref{lem:seg-Idb} yields a segment $(\Gamma, V)$
  such that no $\Delta \in V$ contains $\psi$.
  By induction we then find $(\Gamma, V) \not\Vdash \Diamond\psi$,
  and since $(\Gamma, U) \subseteq (\Gamma, V)$
  this implies $(\Gamma, U) \not\Vdash \Diamond\psi$.
\end{proof}

  We show that the canonical model satisfies~\eqref{eq:Idb-suff}.

\begin{lemma}[\rocqdoc{ComplsegP.general_segP_completeness.html\#CF_Idb}]
  $\mo{X}_P$ satisfies~\eqref{eq:Idb-suff}.
\end{lemma}
\begin{proof}
  Suppose $(\Gamma, U) R (\Delta, V) \subsetsim (\Delta', V')$.
  If $(\Delta', V') = (\expl, \{ \expl \})$ then we can take
  $u = (\expl, \{ \expl \})$ to satisfy~\eqref{eq:Idb-suff}.
  So suppose this is not the case. Then $\bot \notin \Delta'$.
  We claim that
  \begin{equation}\label{eq:Idb-corr-3}
    \Gamma \cup \Diamond(\Delta') \not\vdash \Box((\delta')^c).
  \end{equation}
  If this were false, then using the fact that $\Gamma$ and $\Delta'$ are
  prime we can find $\phi \in \Gamma$, $\delta \in \Delta'$ and $\psi \notin \Delta'$
  such that $\phi \wedge \Diamond\delta \vdash \Box\psi$.
  Then $\phi \vdash \Diamond\delta \to \Box\psi$,
  so $\phi \vdash \Box(\delta \to \psi)$ by $\axref{ax:Idb}$,
  hence $\Box(\delta \to \psi) \in \Gamma$, wherefore $\delta \to \psi \in \Delta \subseteq \Delta'$.
  By assumption $\delta \in \Delta'$, so deductive closure of $\Delta'$
  implies $\psi \in \Delta'$, a contradiction.
  
  So~\eqref{eq:Idb-corr-3} holds, and the Lindenbaum lemma yields a prime theory
  $\Theta$ that contains $\Gamma$ and $\Diamond(\Delta')$ but avoids $\Box((\Delta')^c)$.
  If $\Diamond\bot \in \Theta$ then $\Theta$ is the head of a unique P-segment
  $(\Theta, U')$, and $(\Gamma, U) \subsetsim (\Theta, U') R (\Delta', V')$ by construction.
  Moreover, if $(\Theta, U') \subsetsim (\Gamma'', U'')$ then $\Diamond\bot \in \Gamma''$,
  hence $\expl \in U''$ and $(\expl, \{ \expl \})$ is a world such that
  $(\Gamma'', U'') R (\expl, \{ \expl \})$ and $(\Delta', V') \subsetsim (\expl, \{ \expl \})$.
  So~\eqref{eq:Idb-suff} is satisfied.
  
  If $\Diamond\bot \notin \Theta$, then we take $(\Theta, U_{\Theta,\bot})$.
  By construction $(\Gamma, U) \subsetsim (\Theta, U_{\Theta,\bot}) R (\Delta', V')$.
  Let $(\Gamma'', U'')$ be a P-segment such that $\Theta \subseteq \Gamma''$.
  If $\bot \in \gamma''$ or $\Diamond\bot \in \Gamma''$ then $\expl \in U''$
  and we can use $(\expl, \{ \expl \})$ to witness~\eqref{eq:Idb-suff}.
  If not, then $U''$ must be of the form $U_{\Gamma'',\pi}$ for some $\Diamond\pi \notin \Gamma''$.
  We claim that
  \begin{equation}\label{eq:Idb-corr-2}
    \Box^{-1}(\Gamma'') \cup \Delta' \not\vdash \pi.
  \end{equation}
  If this were false, then there exist $\phi \in \Box^{-1}(\Gamma'')$
  and $\delta \in \Delta'$ such that $\phi \wedge \delta \vdash \pi$.
  This implies $\Box\phi \wedge \Diamond\delta \vdash \Diamond\pi$,
  and since $\Box\phi, \Diamond\delta \in \Gamma''$ we get $\Diamond\pi \in \Gamma''$,
  a contradiction. 
  Using the Lindenbaum lemma,~\eqref{eq:Idb-corr-2} yields a prime theory $\Delta''$
  that we can extend to a segment $(\Delta'', V'')$.
  By construction $\Delta'' \in U''$, so $(\Gamma'', U'') R (\Delta'', V'')$,
  as well as $\Delta' \subseteq \Delta''$.
\end{proof}

\begin{lemma}[\rocqdoc{ComplsegP.CK_Idb_segP_completeness.html\#CKIdb_Strong_Completeness}]
  $\log{CK} \oplus \axref{ax:Idb}$ is sound~and~strongly
  complete with respect to the class of frames satisfying~\eqref{eq:Idb-corr}.
\end{lemma}

  We have now gathered enough background to prove:

\begin{proof}[Proof of Theorem~\ref{thm:strong_compl_for_Ax}]
  Soundness results are straightforwardly obtained as instances of a general result%
  ~(\rocqdoc{Soundness.general_soundness.html\#Soundness})
  via the use of correspondence conditions
  ~(\rocqdoc{Soundness.CK_soundness.html\#CK_Soundness},%
  \rocqdoc{Soundness.CK_Cd_soundness.html\#CKCd_Soundness},%
  \rocqdoc{Soundness.CK_Idb_soundness.html\#CKIdb_Soundness},%
  \rocqdoc{Soundness.WK_soundness.html\#WK_Soundness},%
  \rocqdoc{Soundness.CK_Cd_Idb_soundness.html\#CKCdIdb_Soundness},%
  \rocqdoc{Soundness.CK_Cd_Nd_soundness.html\#CKCdNd_Soundness},%
  \rocqdoc{Soundness.CK_Idb_Nd_soundness.html\#CKIdbNd_Soundness},%
  \rocqdoc{Soundness.IK_soundness.html\#IK_Soundness}).
  For completeness, we consider four cases based on whether or not
  $\axref{ax:Cd}$ and $\axref{ax:Idb}$ are in $\Ax$.
  In each case, adding $\axref{ax:Nd}$ to the logic yields a notion
  of (prime) theory which implies that
  the corresponding canonical model satisfies~\eqref{eq:Nd-corr} or \eqref{eq:Nd-suff}.
  
  \medskip\noindent
  \textit{Case 1: $\axref{ax:Cd} \notin \Ax$ and $\axref{ax:Idb} \notin \Ax$.}
    If $\axref{ax:Nd} \notin \Ax$ then this is Theorem~\ref{thm:ck-completeness}%
    ~(\rocqdoc{Complseg.CK_seg_completeness.html\#CK_Strong_Completeness})
    and, as mentioned, if $\axref{ax:Nd} \in \Ax$ then the canonical model construction
    used in this theorem satisfies~\eqref{eq:Nd-corr}%
    ~(\rocqdoc{Complseg.WK_seg_completeness.html\#WK_Strong_Completeness}).

  \medskip\noindent 
  \textit{Case 2: $\axref{ax:Cd} \in \Ax$ and $\axref{ax:Idb} \notin \Ax$.}
    Use a canonical model based on A- and B-segments
    as outlined above~%
    (\rocqdoc{ComplsegAB.CK_Cd_segAB_completeness.html\#suff_CKCd_Strong_Completeness},%
     \rocqdoc{ComplsegAB.CK_Cd_Nd_segAB_completeness.html\#suff_CKCdNd_Strong_Completeness}).

  \medskip\noindent
  \textit{Case 3: $\axref{ax:Cd} \notin \Ax$ and $\axref{ax:Idb} \in \Ax$.}
    In this case we use a canonical model based on P-segments,
    as outlined above~%
    (\rocqdoc{ComplsegP.CK_Idb_segP_completeness.html\#suff_CKIdb_Strong_Completeness},%
\rocqdoc{ComplsegP.CK_Idb_Ndb_segP_completeness.html\#suff_suff_CKIdbNdb_Strong_Completeness}).
  
  \medskip\noindent
  \textit{Case 4: $\axref{ax:Cd} \in \Ax$ and $\axref{ax:Idb} \in \Ax$.}
    In this final case we can use prime theories as the worlds of our canonical
    models~(\rocqdoc{Complth.general_th_completeness.html\#cworld}), rather than segments. These are partially ordered by inclusion,
    and we define the modal accessibility relation $R$ by
    letting $\Gamma R \Delta$ iff
    $\Box^{-1}(\Gamma) \subseteq \Delta
        \text{ and } \Diamond(\Delta) \subseteq \Gamma$.
    We can view this as the restriction to segments of the
    form $(\Gamma, U)$ where a prime theory $\Delta$ is in $U$ if and only
    if $\Gamma R \Delta$.
    It turns out that $\axref{ax:Cd}$ is crucial to prove the truth lemma,
    while $\axref{ax:Idb}$ allows us to prove that the resulting frame satisfies
    \eqref{eq:Cd-strong} and \eqref{eq:Idb-suffCd}~(\rocqdoc{Complth.general_th_completeness.html\#CF_strong_Cd_weak_Idb}).
    Once again, we obtain strong completeness%
    ~(\rocqdoc{Complth.CK_Cd_Idb_th_completeness.html\#suff_suff_CKCdIdb_Strong_Completeness},%
    \rocqdoc{Complth.IK_th_completeness.html\#suff_suff_suff_IK_Strong_Completeness}).
\end{proof}

  Completeness with respect to correspondence conditions straightforwardly follows. 
  
\begin{corollary}\label{cor:corr_strong_compl_for_Ax}
  Let $\Ax \subseteq \{ \axref{ax:Nd}, \axref{ax:Cd}, \axref{ax:Idb} \}$.
  Then the logic $\log{CK} \oplus \Ax$ is sound and
  strongly complete with respect to the class of frames satisfying
  \textup{($\sf{A}$-corr)} for each $\sf{A} \in \Ax$.
\end{corollary}
\begin{proof}
Note that the cases where $\Ax\subseteq\{\axref{ax:Nd}\}$
are already treated above.
For the remaining ones, as sufficient conditions entail correspondence conditions
(\rocqdoc{Kripke.correspondence.html\#suff_impl_Cd},%
\rocqdoc{Kripke.correspondence.html\#suff_impl_Idb},%
\rocqdoc{Kripke.correspondence.html\#suff_impl_Nd}),
we directly use Theorem~\ref{thm:strong_compl_for_Ax} to
obtain our result 
(\rocqdoc{ComplsegAB.CK_Cd_segAB_completeness.html\#CKCd_Strong_Completeness},%
\rocqdoc{ComplsegAB.CK_Cd_Nd_segAB_completeness.html\#CKCdNd_Strong_Completeness},%
\rocqdoc{ComplsegP.CK_Idb_segP_completeness.html\#CKIdb_Strong_Completeness},%
\rocqdoc{ComplsegP.CK_Idb_Nd_segP_completeness.html\#CKIdbNd_Strong_Completeness},%
\rocqdoc{Complth.CK_Cd_Idb_th_completeness.html\#CKCdIdb_Strong_Completeness},%
\rocqdoc{Complth.IK_th_completeness.html\#IK_Strong_Completeness}).
\end{proof}

%%%%%%%%%%%%%%%%%%%%%%%%%%%%%%%%%%%%%%%%%%%%%%%%%%%%%%%%%%%%%%%%%%%%%%%%%%%%%%%%
\section{Comparison of diamond-free fragments}\label{sec:diamond-free}

  We leverage the sound and complete semantics for extensions of
  $\log{CK}$
  to study their diamond-free fragments.
  We call a logic a \emph{conservative extension of $\log{CK}_{\Box}$}
  if its $\Diamond$-free fragment coincides with $\log{CK}_{\Box}$.
  It is known that the $\Diamond$-free fragment of
  $\log{CK} \oplus \axref{ax:Nd} \oplus \axref{ax:Cd}$ is a conservative
  extension of $\log{CK}_{\Box}$~\cite[Corollary~30]{DasMar23}, and that this is
  not the case for $\log{CK} \oplus \axref{ax:Nd} \oplus \axref{ax:Idb}$.
  In a blog comment~\cite{Das22-comment} Das speculated that ``the real distinction of $\log{IK}$ is due to [\axref{ax:Idb}]'',
  but we falsify this by proving
  that $\log{CK} \oplus \axref{ax:Cd} \oplus \axref{ax:Idb}$
  (hence $\log{CK} \oplus \axref{ax:Idb}$) is conservative over $\log{CK}_{\Box}$.
  As a consequence, we characterise precisely which extensions of $\log{CK}$
  with axioms in $\{ \axref{ax:Nd}, \axref{ax:Cd}, \axref{ax:Idb} \}$
  are conservative over~$\log{CK}_{\Box}$.%
  \footnote{Formally, we show that $\log{CK} \oplus \Ax$ and $\log{CK}$ coincide on their $\Diamond$-free
  fragments for certain $\Ax$, relying on Das and Marin~\cite[Corollary~6]{DasMar23} that this coincides with $\log{CK}_\Box$.}

\begin{proposition}[\rocqdoc{Conservativity.CK_Cd_Idb_conserv_CK.html\#diam_free_eq_CKCdIdb_CK}]\label{prop:k34}
  The logic $\log{CK} \oplus \axref{ax:Cd} \oplus \axref{ax:Idb}$
  is a conservative extension of $\log{CK}_{\Box}$.
\end{proposition}
\begin{proof}
  By definition every formula that is derivable in $\log{CK}_{\Box}$ is also
  derivable in $\log{CK} \oplus \axref{ax:Cd} \oplus \axref{ax:Idb}$,
  so we focus on the converse. We show that for every world
  in every $\log{CK}$-frame we can find a world in a frame for 
  $\log{CK} \oplus \axref{ax:Cd} \oplus \axref{ax:Idb}$ that satisfies
  exactly the same diamond-free formulas.
  This implies that the class of $\log{CK}$-frames validates all diamond-free
  consecutions in $\log{CK} \oplus \axref{ax:Cd} \oplus \axref{ax:Idb}$,
  so that the result follows from completeness of $\log{CK}_{\Box}$ with respect
  to the class of $\log{CK}$-frames.

  Let $\mo{X} = (X, \expl, \leq, R)$ be a $\log{CK}$-frame, and define binary relations
  $\preceq$ and $\mc{R}$ on $X$ by:
  \begin{align*}
    x \preceq y &\iff x \leq y \text{ or } y = \expl  \\
    x \mc{R}  y &\iff x R y  \text{ or } y = \expl
  \end{align*}
  (Note that $\expl$ was already maximal with respect to $\leq$,
  but now it becomes a top element with respect to $\preceq$.)
  Then $\mo{X}' = (X, \expl, \preceq, \mc{R})$ is a $\log{CK}$-frame%
  ~(\rocqdoc{Conservativity.CK_Cd_Idb_conserv_CK.html\#CdIdb_frame_former}).
  The fact that every world can intuitionistically and modally access $\expl$
  implies that $\mo{X}'$ satisfies~\eqref{eq:Cd-strong}%
  ~(\rocqdoc{Conservativity.CK_Cd_Idb_conserv_CK.html\#CdIdb_frame_former_strong_Cd_frame}).
  Moreover, we have that $x \mc{R} y \preceq z$ implies $x \preceq x \mc{R} y \preceq z$ and
  all intuitionistic successors of $x$ can modally access $\expl$, wherefore
  \eqref{eq:Idb-corr} holds%
  ~(\rocqdoc{Conservativity.CK_Cd_Idb_conserv_CK.html\#CdIdb_frame_former_Idb_frame}).
  Therefore $\mo{X}'$ is a frame for $\log{CK} \oplus \axref{ax:Cd} \oplus \axref{ax:Idb}$.

  Valuations for $\mo{X}$ and $\mo{X}'$ coincide
  because we only added relations of the form $x \preceq \expl$ to the
  intuitionistic accessibility relation.
  Let $V$ be such a valuation for $\mo{X}$. Then we can show by induction
  on the structure of $\phi$ that
    \begin{equation*}
    (\mo{X}, V), x \Vdash \phi \iff (\mo{X}', V), x \Vdash \phi  
  \end{equation*}
  for all diamond-free formulas $\phi$%
  ~(\rocqdoc{Conservativity.CK_Cd_Idb_conserv_CK.html\#model_forces_iff_CdIdb_model_forces}).
  This entails that every diamond-free consecution $\Gamma \vdash \phi$ in
  $\log{CK} \oplus \axref{ax:Cd} \oplus \axref{ax:Idb}$ is true
  at $x$. Since $\mo{X}$, $V$ and $x$ are arbitrary, the result follows.
\end{proof}

  We can also give a semantical proof for $\log{CK} \oplus \axref{ax:Nd} \oplus \axref{ax:Cd}$~\cite[Corollary~30]{DasMar23}.
  This uses a different frame transformation:
  whereas the proof of Proposition~\ref{prop:k34} makes $\expl$ modally
  accessible to all worlds to force validity of $\Diamond\bot$,
  the frame transformation in the next proposition falsifies
  $\Diamond\top$ in all worlds.%
  \footnote{Our frame transformation for $\log{CK} \oplus \axref{ax:Cd} \oplus \axref{ax:Idb}$
  tightly connects to the translation given in an independent alternative proof of Proposition~\ref{prop:k34}~\cite{Lang24-blog}
  which maps all diamond formulas to $\bot$. Following this insight, the result for $\log{CK} \oplus \axref{ax:Nd} \oplus \axref{ax:Cd}$
  can be obtained by translating diamond formulas to $\top$ instead~\cite{Das24-comment}.}
  Both approaches ensure validity of
  $\axref{ax:Cd}$ by trivialising diamonds.

\begin{proposition}[\rocqdoc{Conservativity.CK_Cd_Nd_conserv_CK.html\#diam_free_eq_CKCdNd_CK}]\label{prop:k35}
  The logic $\log{CK} \oplus \axref{ax:Cd} \oplus \axref{ax:Nd}$ is
  a conservative extension of $\log{CK}_{\Box}$.
\end{proposition}
\begin{proof}
  We use the same strategy as in the proof of Proposition~\ref{prop:k34}.
  Let $\mo{X} = (X, \expl, \leq, R)$ be any $\log{CK}$-frame.
  We construct a new frame $\mo{X}'$  by adding for each world $x \neq \expl$
  an additional world $x^+$ that satisfies $x \leq x^+ \leq x$
   and that
  cannot modally see anything. Intuitively, this makes $\Diamond\top$ false
  at every world (except $\expl$), forcing distributivity of diamonds
  over joins.
  
  Let $X^+$ be a copy of $X \setminus \{ \expl \}$, and for each $\expl \neq x \in X$
  denote its copy in $X^+$ by $x^+$.%
  \footnote{In the formalisation, $X^+$ is taken to include a copy of $\expl$.
            This copy is forced to be isolated as
            it is only intuitionistically or modally accessible from itself.
            Therefore it does not affect validity in the corresponding model.}
  Define a preorder $\preceq$ on $X \cup X^+$ by
  \begin{align*}
    {\preceq} := \bigcup \big\{ &\{ (x, y), (x, y^+), (x^+, y), (x^+, y^+) \}  \\
                  &\mid x, y \in X \text{ and } x \leq y \big\}.
  \end{align*}
  Leave $R$ unchanged, but view it as a relation on $X \cup X^+$.
  Then $\mo{X}' = (X \cup X^+, \expl, \preceq, R)$ is a $\log{CK}$-frame%
  ~(\rocqdoc{Conservativity.CK_Cd_Nd_conserv_CK.html\#CdNd_frame_former}).
  The worlds $x^+$ ensure satisfaction of~\eqref{eq:Nd-corr}%
  ~(\rocqdoc{Conservativity.CK_Cd_Nd_conserv_CK.html\#CdNd_frame_former_Nd_frame}).
  Furthermore,~\eqref{eq:Cd-corr} is satisfied%
  ~(\rocqdoc{Conservativity.CK_Cd_Nd_conserv_CK.html\#CdNd_frame_former_Cd_frame})%
  , because if
  $x \preceq y$ and $x \preceq z$ then $x^+ \succeq x$ is a world such that
  $R[x^+] \subseteq {\downarrow}R[y]$ and $R[x^+] \subseteq {\downarrow}R[z]$
  (since $R[x^+] = \emptyset$).
  So $\mo{X}'$ is a frame for $\log{CK} \oplus \axref{ax:Cd} \oplus \axref{ax:Nd}$.
  
  Let $V$ be any valuation for $\mo{X}$ and define a valuation $V'$ for
  $\mo{X}'$ by $V'(p) = V(p) \cup \{ x^+ \mid x \in V(p) \}$.
  Then a routine induction on the structure of $\phi$ yields
  $(\mo{X}, V), x \Vdash \phi$ iff $(\mo{X}', V'), x \Vdash \phi$
  for every diamond-free formula $\phi$%
  ~(\rocqdoc{Conservativity.CK_Cd_Nd_conserv_CK.html\#model_forces_iff_CdNd_model_forces}).
  The remainder of the proof is as in Proposition~\ref{prop:k34}.
\end{proof}

  Summarising our results, we get:

\begin{theorem}\label{thm:conservative}
  Let $\Ax \subseteq \{ \axref{ax:Nd}, \axref{ax:Cd}, \axref{ax:Idb} \}$.
  Then $\log{CK} \oplus \Ax$ is a conservative extension of $\log{CK}_{\Box}$ if and only if
  $\axref{ax:Nd} \notin \Ax$ or $\axref{ax:Idb} \notin \Ax$.
\end{theorem}
\begin{proof}
  Non-conservativity of $\{ \axref{ax:Nd}, \axref{ax:Idb} \}$%
  ~(\rocqdoc{Conservativity.CK_Idb_Nd_not_conserv_CK.html\#diam_free_strict_ext_CKIdbNd_CK})
  and $\{ \axref{ax:Nd}, \axref{ax:Cd}, \axref{ax:Idb} \}$%
  ~(\rocqdoc{Conservativity.CK_Idb_Nd_not_conserv_CK.html\#diam_free_strict_ext_IK_CK})
  is witnessed by the fact that they prove $\neg\neg\Box\bot \to \Box\bot$~\cite{DasMar23}, while $\log{CK}$ does not.
  Conservativity of the remaining extensions%
  ~(\rocqdoc{Conservativity.CK_Cd_Idb_conserv_CK.html\#diam_free_eq_CKIdb_CK},%
  \rocqdoc{Conservativity.CK_Cd_Nd_conserv_CK.html\#diam_free_eq_CKCd_CK},%
  \rocqdoc{Conservativity.CK_Cd_Nd_conserv_CK.html\#diam_free_eq_WK_CK})
  follows from
  Propositions~\ref{prop:k34} and~\ref{prop:k35}.
\end{proof}

  Theorem~\ref{thm:conservative} leaves us with two logics that are not
  conservative over $\log{CK}_{\Box}$. While clearly
  $\log{CK} \oplus \axref{ax:Nd} \oplus \axref{ax:Idb}$ is included in $\log{IK}$,
  it is unclear whether or not their diamond-free fragments coincide.

\begin{open}\label{open:1}
  Does the $\Diamond$-free fragment of $\log{CK} \oplus \axref{ax:Nd} \oplus \axref{ax:Idb}$
  coincide with that of $\log{CK} \oplus \axref{ax:Nd} \oplus \axref{ax:Idb}  \oplus \axref{ax:Cd}$?
  Are either or both fragments finitely axiomatisable?
\end{open}

%%%%%%%%%%%%%%%%%%%%%%%%%%%%%%%%%%%%%%%%%%%%%%%%%%%%%%%%%%%%%%%%%%%%%%%%%%%%%%%%
%%%%%%%%%%%%%%%%%%%%%%%%%%%%%%%%%%%%%%%%%%%%%%%%%%%%%%%%%%%%%%%%%%%%%%%%%%%%%%%%
\section{Other axioms}\label{sec:other}

  In Sections~\ref{sec:three-axioms}, \ref{sec:eight} and~\ref{sec:diamond-free}
  we focussed on logics between $\log{CK}$ and $\log{IK}$ generated by
  three axioms. There are of course myriad axioms to be considered,
  which can all be investigated semantically.
  As an example of this, we study two more axioms that have appeared in the literature.

\subsection{The weak normality axiom}

  We first consider the \emph{weak normality axiom}:
  \begin{myenumerate}
    \setlength{\itemindent}{2em}
    \myitem{$\mathsf{N_{\Diamond\Box}}$} \label{ax:Ndb}
          $\Diamond\bot \to \Box\bot$
  \end{myenumerate}
  This axiom was used by Kojima~\cite{Koj12} in the context of neighbourhood
  semantics, weakening $\log{WK}$ by replacing $\axref{ax:Nd}$
  with $\axref{ax:Ndb}$. We investigate how this axiom fits in our semantic
  framework, and use it to compare extensions
  of $\log{CK}$ that include $\axref{ax:Ndb}$. This gives
  rise to logics whose diamond-free fragments do not coincide with the ones
  from Section~\ref{sec:diamond-free}.

\begin{proposition}
  A frame validates $\axref{ax:Ndb}$ whenever it satisfies%
  ~$(\rocqdoc{Kripke.correspondence.html\#sufficient_Ndb})$:
  \begin{myenumerate}
    \setlength{\itemindent}{3.3em}
    \myitem{\ref{ax:Ndb}-suff} \label{eq:Ndb-suff}
      $\forall x$ (if $xR\expl$ and $xRy$ then $y = \expl$)
  \end{myenumerate}
  Moreover, a frame validates $\axref{ax:Ndb}$ if and only if it satisfies%
  ~$(\rocqdoc{Kripke.correspondence.html\#correspond_Ndb})$:
  \begin{enumerate}
    \setlength{\itemindent}{3.3em}
    \myitem{\ref{ax:Ndb}-corr} \label{eq:Ndb-corr}
      $\begin{aligned}[t]
        \forall x (
          &\text{if } \forall y \; (x \leq y \text{ implies } yR\expl ) \\
          &\text{then } \forall y, z \; (x \leq y R z \text{ implies } z = \expl))
       \end{aligned}$
  \end{enumerate}
\end{proposition}
\begin{proof}
  The first claim follows from the second, because~\eqref{eq:Ndb-suff}
  implies~\eqref{eq:Ndb-corr} (\rocqdoc{Kripke.correspondence.html\#suff_impl_Ndb}).
  The second claim follows from a routine verification. 
\end{proof}

  What happens if we add $\axref{ax:Ndb}$ to any of the
  eight logical systems discussed in Section~\ref{sec:three-axioms}?
  Since $\axref{ax:Nd}$ implies $\axref{ax:Ndb}$, adding it to a system
  containing $\axref{ax:Nd}$ does not change anything.
  Furthermore, Theorem~\ref{thm:conservative} implies that Kojima's logic $\log{CK} \oplus \axref{ax:Ndb}$,
  as well as its extension with $\axref{ax:Cd}$,
  are conservative over $\log{CK}_{\Box}$%
  ~(\rocqdoc{Conservativity.CK_Cd_Nd_conserv_CK.html\#diam_free_eq_CKNdb_CK},%
  \rocqdoc{Conservativity.CK_Cd_Nd_conserv_CK.html\#diam_free_eq_CKCdNdb_CK}).
  This leaves us with two new logics (in terms of their diamond-free fragment)
  between $\log{CK}_{\Box}$ and $\log{IK}$. We claim that they relate as in
  Figure~\ref{fig:Ndb}.
  \begin{figure*}
    \vspace{0em}
    \centering
    \begin{tikzpicture}[yscale=.7,xscale=1.1]
      %% nodes
      \node (ik) at (-1,0) {$\log{CK}_{\Box}$};
      \node (NdbI) at (2,0) {$\log{CK} \oplus \axref{ax:Ndb} \oplus \axref{ax:Idb}$};
      \node (NdI) at (6,1) {$\log{CK} \oplus \axref{ax:Nd} \oplus \axref{ax:Idb}$};
      \node (NdbCI) at (6,-1) {$\log{CK} \oplus \axref{ax:Ndb} \oplus \axref{ax:Cd} \oplus \axref{ax:Idb}$};
      \node (IK) at (10,0) {$\log{IK}$};
      %% edges
      \draw[latex-, dashed] (ik) to (NdbI);
      \draw[bend left=15, latex-, dashed] (NdbI) to (NdI);
      \draw[bend left=15, latex-] (NdI) to node[above]{?} (IK);
      \draw[-latex, dashed, bend right=25] (NdI) to (NdbCI);
      \draw[-latex, dashed,bend left=15] (IK) to (NdbCI);
      \draw[-latex, bend left=15] (NdbCI) to node[below]{?} (NdbI);
      \draw[-latex, bend right=25] (NdbCI) to node[right]{?} (NdI);
    \end{tikzpicture}
    \vspace{0em}
    \caption{Logics including $\axref{ax:Ndb}$.
      Inclusion of logics from left to right is immediate.
      Dashed arrows indicate non-inclusion of diamond-free fragments,
      and follow from Propositions~\ref{prop:nnBp-Bnnp} and~\ref{app:prop:CK-vs-CKNdbI}.
      We note that non-inclusion of the middle question mark would imply
      non-inclusion of the other two.}
    \label{fig:Ndb}
  \end{figure*}
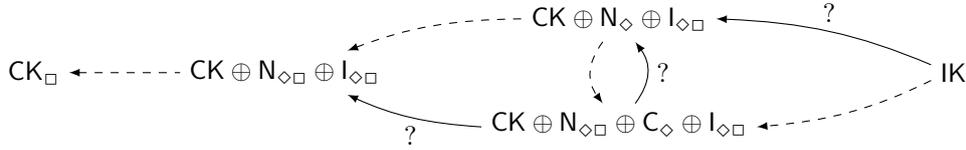

\begin{theorem}
  Let $\Ax \subseteq \{ \axref{ax:Cd}, \axref{ax:Idb} \}$.
  Then the logic $\log{CK} \oplus \axref{ax:Ndb} \oplus \Ax$ is sound and
  strongly complete with respect to the class of frames satisfying
  \eqref{eq:Ndb-corr}, and \textup{($\sf{A}$-suff)} for each $\sf{A} \in \Ax$%
  ~$(\rocqdoc{Complseg.CK_Ndb_seg_completeness.html\#CKNdb_Strong_Completeness},%
  \rocqdoc{ComplsegAB.CK_Cd_Ndb_segAB_completeness.html\#suff_CKCdNdb_Strong_Completeness})$.
  Moreover, if $\axref{ax:Idb}\in \Ax$ then we can replace
  \eqref{eq:Ndb-corr} by \eqref{eq:Ndb-suff}%
  ~$(\rocqdoc{ComplsegP.CK_Idb_Ndb_segP_completeness.html\#suff_suff_CKIdbNdb_Strong_Completeness},%
    \rocqdoc{Complth.CK_Cd_Idb_Ndb_th_completeness.html\#suff_suff_suff_CKCdIdbNdb_Strong_Completeness})$.
\end{theorem}
\begin{proof}
  Here again soundness is straightforward
  ~(\rocqdoc{Soundness.CK_Ndb_soundness.html\#CKNdb_Soundness},%
  \rocqdoc{Soundness.CK_Cd_Ndb_soundness.html\#CKCdNdb_Soundness},%
  \rocqdoc{Soundness.CK_Idb_Ndb_soundness.html\#CKIdbNdb_Soundness},%
  \rocqdoc{Soundness.CK_Cd_Idb_Ndb_soundness.html\#CKCdIdbNdb_Soundness}).
  Completeness follows from verifying that the canonical model constructed for each
  of the resulting logics satisfies~\eqref{eq:Ndb-corr} or~\eqref{eq:Ndb-suff}.
\end{proof}

  The next propositions witness the non-inclusions indicated in Figure~\ref{fig:Ndb}.

\begin{proposition}[\rocqdoc{Conservativity.CK_Idb_Nd_not_incl_CK_Cd_Idb_Ndb.html\#CKIdbNd_not_included_CKCdIdbNdb}]\label{prop:nnBp-Bnnp}
  The formula $\neg\neg\Box p \to \Box\neg\neg p$
  is derivable in $\log{CK} \oplus \axref{ax:Nd} \oplus \axref{ax:Idb}$
  but not in $\log{CK} \oplus \axref{ax:Ndb} \oplus \axref{ax:Cd} \oplus \axref{ax:Idb}$
  (hence not in $\log{CK} \oplus \axref{ax:Ndb} \oplus \axref{ax:Idb}$).
\end{proposition}
\begin{proof}
  Derivability in $\log{CK} \oplus \axref{ax:Nd} \oplus \axref{ax:Idb}$
  is a known result~\cite[Lemma 10]{DasMar23}.
  Figure~\ref{fig:app-Ndb-2} gives a frame that satisfies
  \eqref{eq:Ndb-suff}, \eqref{eq:Cd-strong} and~\eqref{eq:Idb-suff}
  where the formula is not valid (it is false at $x$
  if we take a valuation with $V(p) = \{ \expl \}$),
  hence it is not derivable in
  $\log{CK} \oplus \axref{ax:Ndb} \oplus \axref{ax:Idb}$.
\end{proof}

\begin{proposition}[\rocqdoc{Conservativity.CK_Idb_Ndb_not_conserv_CK.html\#diam_free_strict_ext_CKIdbNdb_CK}]\label{app:prop:CK-vs-CKNdbI}
  The formula
  \begin{equation}\label{eq:Ndb-Idb-form}
    \neg\Box\bot \to (\neg\neg\Box p \to \Box\neg\neg p).
  \end{equation}
  is derivable in $\log{CK} \oplus \axref{ax:Ndb} \oplus \axref{ax:Idb}$,
  but not in $\log{CK}$.
\end{proposition}
\begin{proof}
  To see that the formula is not derivable in $\log{CK}$,
  consider the frame from Figure~\ref{fig:Ndb-1} with valuation
  $V(p) = \{ w, \expl \}$.
  Then clearly $x$ and $y$ do not satisfy $\Box\bot$, hence
  $x \Vdash \neg\Box\bot$.
  Furthermore, the fact that $y \Vdash \Box p$ implies that
  neither $x$ nor $y$ satisfies $\neg\Box p$, and hence
  $x \Vdash \neg\neg\Box p$.
  Finally, $z \not\Vdash \neg\neg p$ because $z \Vdash \neg p$,
  which implies $x \not\Vdash \Box\neg\neg p$.
  Combining this entails that $x$ falsifies the formula.

  Next, we derive the formula in $\log{CK} \oplus \axref{ax:Ndb} \oplus \axref{ax:Idb}$.
  Applying $\axref{rule:Nec}$, \axref{ax:Kb} and $\axref{rule:MP}$ to $p \to \neg\neg p$, a theorem of IPL,
  yields $\Box p \to \Box\neg\neg p$.
  This is equal to $\Box p \to \Box(\neg p \to \bot)$,
  so \axref{ax:Kd} yields $\Box p \to (\Diamond\neg p \to \Diamond\bot)$.
  Using $\axref{ax:Ndb}$ gives $\Box p \to (\Diamond\neg p \to \Box\bot)$.
  Using propositional intuitionistic reasoning we rewrite this to
  \begin{equation*}
    \neg\neg\Box p \to (\neg\Box\bot \to \neg\Diamond\neg p).
  \end{equation*}
  Now currying and commutativity of $\wedge$ allows us to derive
  $\neg\Box\bot \to (\neg\neg\Box p \to \neg\Diamond\neg p)$.
  Rewriting $\neg\Diamond\neg p$ to $\Diamond\neg p \to \bot$
  and using the fact that $\bot \to \Box\bot$ yields
  \begin{equation*}
    \neg\Box\bot \to (\neg\neg\Box p \to (\Diamond\neg p \to \Box\bot)).
  \end{equation*}
  Finally, applying~$\axref{ax:Idb}$ we find
  $\neg\Box\bot \to (\neg\neg\Box p \to \Box \neg\neg p)$.
\end{proof}

%% Figures accompanying the propositions
\begin{figure}[h!]
  \vspace{-1em}
  \begin{subfigure}{.22\textwidth}
    \centering
    \begin{tikzpicture}[scale=.6]
      \draw[rounded corners, fill=yellow!10, semithick] (5.6,2.8) rectangle (10.4,-.8);
        %% nodes
        \node (x2) at (6.5,0) {$x$};
        \node (y2) at (6.5,2) {$y$};
        \node (z2) at (9.5,0) {$z$};
        \node (e2) at (9.5,2) {$\expl$};
        %% edges
        \draw[-latex] (x2) \ito{left} (y2);
        \draw[-latex] (z2) \ito{right} (e2);
        \draw[-latex] (x2) \rto{below} (z2);
        \draw[-latex] (y2) \rto{above} (e2);
    \end{tikzpicture}
    \caption{Frame for Prop.~\ref{prop:nnBp-Bnnp}.}
    \label{fig:app-Ndb-2}
  \end{subfigure}
  \begin{subfigure}{.25\textwidth}
    \begin{tikzpicture}[scale=.6]
      \draw[rounded corners, fill=yellow!10, semithick] (5.6,2.8) rectangle (13.2,-.8);
        %% nodes
        \node (x2) at (6.5,0) {$x$};
        \node (y2) at (6.5,2) {$y$};
        \node (z2) at (9.5,0) {$z$};
        \node (w2) at (9.5,2) {$w$};
        \node (e2) at (11,1)  {$\expl$};
        %% edges
        \draw[-latex] (x2) \ito{left} (y2);
        \draw[-latex] (x2) \rto{below} (z2);
        \draw[-latex] (y2) \rto{above} (w2);
        \draw[-latex] ([shift=(135:.6)]11.6,1) arc (135:-135:.6) node[right,pos=.5]{\footnotesize{$R$}};
    \end{tikzpicture}
    \caption{Frame for Proposition~\ref{prop:wCDb}.}
    \label{fig:Ndb-1}
  \end{subfigure}
  \caption{Two more frames used to falsify formulas.}
  \label{fig:app-Ndb-frames}
\end{figure}
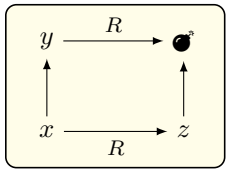
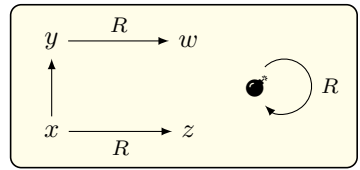

  We have the following analogue of Open question~\ref{open:1}.

\begin{open}
  Does the $\Diamond$-free fragment of $\log{CK} \oplus \axref{ax:Ndb} \oplus \axref{ax:Idb}$
  coincide with that of $\log{CK} \oplus \axref{ax:Ndb} \oplus \axref{ax:Idb}  \oplus \axref{ax:Cd}$?
  Are the $\Diamond$-free fragments of either of these logics finitely axiomatisable?
\end{open}

%%%%%%%%%%%%%%%%%%%%%%%%%%%%%%%%%%%%%%%%%%%%%%%%%%%%%%%%%%%%%%%%%%%%%%%%%%%%%%%%
\subsection{The weak constant domain axiom}\label{subsec:wCD}

  The logic $\log{FIK}$ arises from $\log{IK}$ by replacing
  $\axref{ax:Idb}$ with the \emph{weak constant domain} axiom:
  \begin{myenumerate}
    \setlength{\itemindent}{2em}
    \myitem{$\mathsf{wCD}$} \label{ax:wCD}
          $\Box(p \vee q) \to ((\Diamond p \to \Box q) \to \Box q)$
  \end{myenumerate}
  It was introduced by Balbiani, Gao, Gencer, and Olivetti~\cite{BalGaoGenOli24} to axiomatise the
  intuitionistic modal logic with birelational semantics satisfying
  \eqref{eq:Cd-strong}, but not necessarily~\eqref{eq:Idb-suffCd}.
  We briefly investigate $\axref{ax:wCD}$ and how it relates to the
  axioms from Section~\ref{sec:three-axioms} to
  obtain results similar to those of Figure~\ref{fig:Ndb}.
  We postpone completeness results and finite axiomatisability of some
  of the diamond-free fragments to future work.

\begin{proposition}[\rocqdoc{GHC.properties.html\#CKIdb_prv_wCD}]
  Over $\log{CK}$, the axiom $\axref{ax:Idb}$ implies $\axref{ax:wCD}$.
\end{proposition}
\begin{proof}
  Suppose we have $(\Diamond p \to \Box q) \to \Box(p \to q)$.
  Then it follows that
  \begin{equation*}
    \Box(p \vee q) \to ((\Diamond p \to \Box q) \to (\Box(p \vee q) \wedge \Box(p \to q))).
  \end{equation*}
  The conjunction is equal to $\Box((p \vee q) \wedge (p \to q))$,
  and since $((p \vee q) \wedge (p \to q)) \to q$ is a theorem of
  intuitionistic logic it implies $\Box q$.
  Thus we find $\Box(p \vee q) \to ((\Diamond p \to \Box q) \to \Box q$.
\end{proof}

  The previous proposition implies that adding $\axref{ax:wCD}$ to a logic
  that already proves $\axref{ax:Idb}$ does not add any extra deductive power.
  So adding $\axref{ax:wCD}$ to the logics from
  Section~\ref{sec:three-axioms} yields four logics of interest:
  $\log{CK} \oplus \axref{ax:wCD}$, $\log{CK} \oplus \axref{ax:Cd} \oplus \axref{ax:wCD}$,
  $\log{CK} \oplus \axref{ax:Nd} \oplus \axref{ax:wCD}$, and $\log{FIK}$.
  The diamond-free fragments of the first two coincide with $\log{CK}_{\Box}$
  as a consequence of Theorem~\ref{thm:conservative}. Thus we are left 
  with a situation similar to Figure~\ref{fig:Ndb}, depicted in Figure~\ref{fig:wCD}.

  \begin{figure*}
    \centering
    \begin{tikzpicture}[yscale=.7,xscale=1.1]
      %% nodes
      \node (ik) at (-1.5,0) {$\log{CK}_{\Box}$};
      \node (NdW) at (2,0) {$\log{CK} \oplus \axref{ax:Nd} \oplus \axref{ax:wCD}$};
      \node (NdI) at (6,1) {$\log{CK} \oplus \axref{ax:Nd} \oplus \axref{ax:Idb}$};
      \node (NdCW) at (6,-1) {$\log{CK} \oplus \axref{ax:Nd} \oplus \axref{ax:Cd} \oplus \axref{ax:wCD}$};
      \node (IK) at (10,0) {$\log{IK}$};
      %% edges
      \draw[latex-, dashed] (ik) to (NdW);
      \draw[bend left=15, latex-, dashed] (NdW) to (NdI);
      \draw[bend left=15, latex-] (NdI) to node[above]{?} (IK);
      \draw[-latex, dashed, bend right=25] (NdI) to (NdCW);
      \draw[-latex, dashed,bend left=15] (IK) to (NdCW);
      \draw[-latex, bend left=15] (NdCW) to node[below]{?} (NdW);
      \draw[-latex, bend right=25] (NdCW) to node[right]{?} (NdI);
    \end{tikzpicture}
    \vspace{0em}
    \caption{Logics including $\axref{ax:wCD}$.
      Inclusion of logics from left to right is immediate.
      Dashed arrows indicate non-inclusion of diamond-free fragments,
      and follow from Propositions~\ref{prop:CKNI-FIK} and~\ref{prop:wCDb}.}
    \label{fig:wCD}
  \end{figure*}
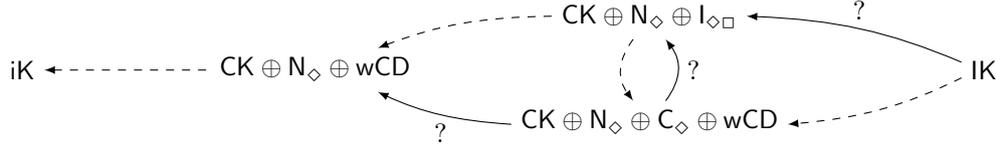

\begin{proposition}\label{prop:CKNI-FIK}
  The formula 
  $\neg\neg\Box p \to \Box\neg\neg p$ is derivable in
  $\log{CK} \oplus \axref{ax:Nd} \oplus \axref{ax:Idb}$~$(\rocqdoc{GHC.properties.html\#negneg_box_prv})$ but not in $\log{FIK}$.
\end{proposition}
\begin{proof}
  The first part of the statement is a known result~\cite[Lemma 10]{DasMar23}.
  The second part follows from constructing a countermodel in the semantics
  for $\log{FIK}$~\cite{BalGaoGenOli24}.
\end{proof}

\begin{proposition}[\rocqdoc{Conservativity.CK_wCD_Nd_not_conserv_CK.html\#diam_free_strict_ext_CKwCDNd_CK}]\label{prop:wCDb}
  The formula 
  \begin{myenumerate}
    \setlength{\itemindent}{2em}
    \myitem{$\mathsf{wCD_{\Box}}$} \label{ax:wCDb}
      $\Box(p \vee q) \to ((\neg\Box\neg p \to \Box q) \to \Box q)$
  \end{myenumerate}
  is derivable in
  $\log{CK} \oplus \axref{ax:Nd} \oplus \axref{ax:wCD}$ but not in $\log{CK}$.
\end{proposition}
\begin{proof}
  The formula $\Diamond p \to \neg\Box\neg p$ is a theorem of
  $\log{CK} \oplus \axref{ax:Nd}$.
  Now using instances of the intuitionistic theorems
  \begin{align*}
    &(\phi \to \psi) \to ((\chi \to \phi) \to (\chi \to \psi)) \\
    \quad\text{and}\quad
    &(\phi \to \psi) \to ((\psi \to \chi) \to (\phi \to \chi))
  \end{align*}
  we can derive $\axref{ax:wCDb}$ from $\axref{ax:wCD}$.
  The fact that $\axref{ax:wCDb}$ is not derivable in $\log{CK}$ can be
  shown by taking the $\log{CK}$-frame from Figure~\ref{fig:Ndb-1} and setting
  $V(p) = \{ z, \expl \}$ and $V(q) = \{ w, \expl \}$.
  The resulting model falsifies $\axref{ax:wCDb}$ at $x$.
\end{proof}

\begin{open}
  Does the $\Diamond$-free fragment of $\log{CK} \oplus \axref{ax:Nd} \oplus \axref{ax:wCD}$
  coincide with that of $\log{CK} \oplus \axref{ax:Nd} \oplus \axref{ax:wCD}  \oplus \axref{ax:Cd}$?
  Are the $\Diamond$-free fragments of either of these logics finitely axiomatisable?
\end{open}

%%%%%%%%%%%%%%%%%%%%%%%%%%%%%%%%%%%%%%%%%%%%%%%%%%%%%%%%%%%%%%%%%%%%%%%%%%%%%%%%
%%%%%%%%%%%%%%%%%%%%%%%%%%%%%%%%%%%%%%%%%%%%%%%%%%%%%%%%%%%%%%%%%%%%%%%%%%%%%%%%
\section{A note on the formalisation}\label{sec:formal}

  Beyond the extra confidence it provides in our results,
  for example clarifying the semantics of $\log{WK}$ which have been disputed~\cite{MendeP05},
  our use of formalisation in the Rocq Prover has been crucial for our work.
  
  Because of the multitude of axiomatic extensions of $\log{CK}$
  we consider, we had to formalise most of our results
  parametrised by an arbitrary set of axioms.
  This generality led to formalised results which are both
  insightful and close to their pen-and-paper counterparts by their
  size and structure.
  For example, while it is often questioned in modal logic~\cite{HakNeg12},
  we proved the deduction theorem for \emph{all} axiomatic extensions of $\log{CK}$.
  Additionally, each of our strong completeness results via a canonical model 
  construction, roughly 500 lines of code, was leveraged
  for several logics, around 50 lines only per logic.
  
  Pragmatically, the code for a certain set of logics
  could be copied, pasted, and modified for another such set
  in a heuristic and efficient way. 
  This enabled us to study a total of 12 logics,
  with an additional few days of work for each further
  axiom analysed.
  We invite our readers to experience it for themselves
  by downloading our code and adding their favourite axioms.

%%%%%%%%%%%%%%%%%%%%%%%%%%%%%%%%%%%%%%%%%%%%%%%%%%%%%%%%%%%%%%%%%%%%%%%%%%%%%%%%
%%%%%%%%%%%%%%%%%%%%%%%%%%%%%%%%%%%%%%%%%%%%%%%%%%%%%%%%%%%%%%%%%%%%%%%%%%%%%%%%
\section{Further work}\label{sec:further}

We have shown that using a relational semantics for $\log{CK}$, and employing Rocq to verify our proofs and tame the profusion of logics arising from
combinations of axioms, allows us to analyse the modal logics between $\log{CK}$ and $\log{IK}$. The success of this methodology leaves us
with much further work to pursue.

\paragraph{Open questions} While our work has closed some open questions about the logics between $\log{CK}$ and $\log{IK}$, others remain.
  Most pressingly, do the $\Diamond$-free fragments of $\log{CK}\oplus\axref{ax:Nd}\oplus\axref{ax:Idb}$ and $\log{IK}$ coincide, and do they have
  a finite axiomatisation? We are aware of a discussion, reported by Das and Marin~\cite{DasMar23}, that the fragment of $\log{IK}$ might not
  have a finite axiomatisation, but we are not aware of any proof of this; Grefe's thesis~\cite{Gre99} proves only the much weaker result that there
  exists an infinite chain of finitely axiomatisable $\Diamond$-free logics between $\log{CK}$ and $\log{IK}$.

\paragraph{Finite model properties}
        As noted in Remark~\ref{rem:CK-fmp}, the canonical model construction
        for $\log{CK}$ can be carried out relative to a set $\Sigma$ of
        formulas. This yields the finite
        model property for $\log{CK}$, as proved already~\cite[Section~4]{MendeP05}.
        Since the completeness results in Theorem~\ref{thm:completeness} use
        various model constructions, the finite model property does not
        readily carry over. This raises the question whether we can
        obtain finite model properties for the logics
        between $\log{CK}$ and $\log{IK}$.

\paragraph{Axiomatic extensions} Focus on the $\log{K}$ family of logics is reasonable to answer basic questions in modal logic, but much
  interest concerns extensions of a basic modal logic with further axioms, from $\log{S4}$ to provability logic to epistemic logic. How does the choice of base
  logic, from $\log{CK}$ to $\log{IK}$, effect the properties of these logics?
  Adding even more axioms into the mix will make our
  Rocq framework more crucial than ever.
  
\paragraph{Weaker or incomparable basic logics}
  While a lot of the work in intuitionistic modal logic builds on at least $\log{CK}$,
  this is not universal. For example, Bo\v{z}i\'{c} and Do\v{s}en~\cite{BozDos84},
  Wolter and Zakharyaschev~\cite{WolZak98,WolZak99}, and Gor\'{e}, Postniece and Tiu~\cite{GorPosTiu10} take as
  basis necessitation, \axref{ax:Kb}, \axref{ax:Nd}, and \axref{ax:Cd}, which is
  incomparable with $\log{CK}$ because \axref{ax:Kd} is merely an optional extension.
  This setup, with no assumed links between $\Box$ and $\Diamond$, is too weak
  for a sensible birelational semantics, so instead takes a \emph{tri}relational semantics with separate relations for the two modalities. However if we
  wished to remain within birelational semantics we could still consider logics where $\Box$ and $\Diamond$ interact, but not via \axref{ax:Kd}, or
  logics with weaker properties of $\Box$, forgoing necessitation or \axref{ax:Kb}.

\paragraph{Sahlqvist theorems} Finding a link between axioms and relational properties can take ingenuity. Sahlqvist theorems for classical modal
  logic~\cite{Sah75,SamVac89,ForMor23} make this automatic for certain syntactically defined classes of formulas. If these techniques
  could be modified to take the relational semantics for $\log{CK}$ as their basis, they would provide another tool for taming the jungle of
  logics in this space.
  
\paragraph{Proof theory and types} There is a body of research too large to summarise on proof theory for intuitionistic modal logics, but our work
  may help to identify new logics to target and techniques to use, such as labelled sequents that take $\log{CK}$-semantics as their basis. In
  particular, much~recent~work on proofs for intuitionistic modal logic involves calculi with type assignment, but no $\Diamond$. These
  type theories support impressive applications in formalised mathematics, so it may be worthwhile to investigate type theory with
  $\Diamond$ from a $\log{CK}$ base, following Bellin, De Paiva and Ritter~\cite{BeldePRit01}, while simultaneously taking cues from more
  recent work on modal type theory.

\bibliographystyle{IEEEtranS}
\bibliography{modal-int}

\end{document}